\documentclass[10pt]{article}
\usepackage{amsmath}
\usepackage{amsthm}
\usepackage{amsfonts}
\usepackage{amssymb}
\usepackage{url}
\usepackage{hyperref}
\usepackage{eucal}
\usepackage{mathrsfs}
\usepackage{graphicx,graphics}
\usepackage[T1]{fontenc}%
\usepackage{url}

\newtheorem{lemma}{Lemma}[section]
\newtheorem{theorem}{Theorem}[section]

\newcommand\fig[1] {{\rm Figure}~\ref{fig:#1}}
\newcommand\labfig[1] {\label{fig:#1}}
\newcommand\sect[1] {\ref{sect:#1}}
\newcommand\labsect[1] {\label{sect:#1}}

\newcommand{\bfm}[1]{\mbox{\boldmath ${#1}$}}
\newcommand{\nonum}{\nonumber \\}
\newcommand\eq[1] {(\ref{#1})}
\newcommand{\beqa}{\begin{eqnarray}}
\newcommand{\eeqa}[1]{\label{#1}\end{eqnarray}}
\newcommand{\beq}{\begin{equation}}
\newcommand{\eeq}[1]{\label{#1}\end{equation}}
\newcommand{\R}{\mathbb{R}}

\newcommand{\Grad}{\nabla}
\newcommand{\Div}{\nabla \cdot}

\newcommand{\Tr}{\mathop{\rm Tr}\nolimits}

\newcommand{\lang}{\langle}
\newcommand{\rang}{\rangle}
\newcommand{\Md}{\partial}

\newcommand{\dund}[1]{\underline{\underline{#1}}}
\newcommand{\und}[1]{\smash{\underline{#1}}}

\newcommand{\Ga}{\alpha}
\newcommand{\Gb}{\beta}
\newcommand{\Gd}{\delta}
\newcommand{\Ge}{\epsilon}

\newcommand{\Gg}{\gamma}

\newcommand{\Gk}{\kappa}

\newcommand{\Gl}{\lambda}
\newcommand{\Gn}{\eta}
\newcommand{\Gm}{\mu}
\newcommand{\Gv}{\nu}

\newcommand{\Gt}{\theta}

\newcommand{\Gz}{\zeta}

\newcommand{\GP}{\Pi}

\newcommand{\GO}{\Omega}

\newcommand{\BGe}{\bfm\epsilon}

\newcommand{\BGn}{\bfm\eta}

\newcommand{\BGs}{\bfm\sigma}

\newcommand{\BGx}{\bfm\xi}

\newcommand{\BGG}{\bfm\Gamma}
\newcommand{\BGL}{\bfm\Lambda}

\newcommand{\CA}{{\cal A}}

\newcommand{\CV}{{\cal V}}

\newcommand{\bpm}{\begin{pmatrix}}
\newcommand{\epm}{\end{pmatrix}}

\def\b0{\bf 0}

\def\Ba{{\bf a}}
\def\Bb{{\bf b}}
\def\Bc{{\bf c}}
\def\Bd{{\bf d}}
\def\Be{{\bf e}}

\def\Bk{{\bf k}}

\def\Bm{{\bf m}}
\def\Bn{{\bf n}}

\def\Bt{{\bf t}}
\def\Bu{{\bf u}}
\def\Bv{{\bf v}}
\def\Bw{{\bf w}}
\def\Bx{{\bf x}}
\def\By{{\bf y}}
\def\Bz{{\bf z}}
\def\BA{{\bf A}}
\def\BB{{\bf B}}
\def\BC{{\bf C}}

\def\BI{{\bf I}}

\def\BM{{\bf M}}
\def\BN{{\bf N}} 
\def\BO{{\bf O}}
\def\BP{{\bf P}}

\def\BS{{\bf S}}

\title{On the possible effective elasticity tensors of 2-dimensional and 3-dimensional printed materials}
\date{}
\begin{document}
\maketitle
\vskip -.5cm
\centerline{\large
Graeme W. Milton\footnote{Department of Mathematics, University of Utah, USA -- milton@math.utah.edu,},
\quad
Marc Briane\footnote{Institut de Recherche Math\'ematique de Rennes, INSA de Rennes, FRANCE -- mbriane@insa-rennes.fr,}
\,\, and \,\,
Davit Harutyunyan \footnote{Department of Mathematics, University of Utah, USA -- davith@math.utah.edu,},
}
\vskip 1.cm
\begin{abstract}
The set $GU_f$ of possible effective elastic tensors of composites built from two materials with elasticity tensors $\BC_1>0$ and $\BC_2=0$
comprising the set $U=\{\BC_1,\BC_2\}$ and mixed in proportions $f$ and $1-f$ is 
partly characterized. The material with tensor $\BC_2=0$ corresponds to
a material which is void. (For technical reasons $\BC_2$ is actually taken to be nonzero and we take the limit $\BC_2\to 0$).
Specifically, recalling that $GU_f$ is completely characterized through minimums of sums of energies, involving a set of
applied strains, and complementary energies, involving a set of applied stresses, we provide descriptions of 
microgeometries that in appropriate limits achieve the minimums in many cases. In these cases the calculation of the minimum is reduced to a finite dimensional
minimization problem that can be done numerically. Each microgeometry consists of a union of walls in appropriate directions, where the material
in the wall is an appropriate $p$-mode material, that is easily compliant to $p\leq 5$ independent applied strains, yet supports any stress
in the orthogonal space. Thus the material can easily slip in certain directions along the walls. The region outside the walls contains ``complementary Avellaneda material'' which
is a hierarchical laminate which minimizes the sum of complementary energies. 
\end{abstract}

\section{Introduction}
\setcounter{equation}{0}
Here we consider what effective elasticity tensors can be produced in the limit $\Gd\to 0$ if we mix in prescribed proportions two materials with
positive definite and bounded elasticity tensors $\BC_1$ and $\BC_2=\Gd\BC_0$. In the limit $\Gd\to 0$ this represents a mixture of an elastic phase and
an extremely compliant phase. Thus we are given a set $U=\{\BC_1,\Gd\BC_0\}$ and we are aiming to characterize as best
we can the set $GU_f$ representing the set of all possible effective tensors of composites having a volume fraction $f$ of phase 1. The elasticity tensor
 $\BC_1$ need not be isotropic but if it is anisotropic we require that it has a fixed orientation throughout the composite.   Our results are summarized by the theorem in the conclusion section.

 To get an idea of the enormity
of the problem one has to recognize that in three-dimensions elasticity tensors can be represented by $6\times 6$ matrices and these have $21$ independent elements. 
The set of possible elasticity tensors is thus represented as a set in a $21$ dimensional space. Even a distorted multidimensional cube in a $21$-dimensional space needs about
44 million real numbers to represent it (specifying the position in $21$-dimensional space of each of the $2^{21}$ vertices). In the case the two phases are isotropic, one is free to rotate the material to obtain an equivalent
structure. Thus the set of possible elasticity tensors is invariant under rotations. As rotations involve $3$ parameters
(the Euler angles) this reduces the number of constants needed to describe the elasticity tensor from $21$ to $21-3=18$, and thus the elasticity tensor can be represented in an $18$-dimensional space of tensor invariants. For example,
in the generic case, one can take these 18 invariants as follows: the 6 eigenvalues of the elasticity tensor; the 2 
independent elements of the normalized eigenstrain associated with the lowest eigenvalue, that can be assumed to be diagonal by an appropriate choice of the coordinate axes (which then fixes these axes); the 4 independent elements
of the normalized eigenstrain associated with the second lowest eigenvalue that is orthogonal to the first 
eigenstrain; the 3 independent elements of the normalized eigenstrain associated with the third lowest eigenvalue that is orthogonal to the first two eigenstrains; the 2 independent elements of the normalized eigenstrain associated with the third lowest eigenvalue that is orthogonal to the first three eigenstrains; and the one independent elements of the normalized eigenstrain associated with the third lowest eigenvalue that is orthogonal to the first four eigenstrains. This brings the total to 6+2+4+3+2+1=18. In the same way that it takes two parameters (the bulk and shear moduli) to specify the elastic behavior of an isotropic material, it takes 18 parameters to specify the elastic behavior of a fully anisotropic material.

A distorted cube in this $18$- dimensional space still requires about 4.7 million numbers to represent it. This makes exploring the range of possible elasticity tensors a daunting, if not impossible, numerical task. Some numerical exploration of this space
has been done by Sigmund \cite{Sigmund:1994:MPC,Sigmund:1995:TMP}, but we emphasize that this exploration covers only a tiny fraction of the number of possibilities. 

Furthermore, the microstructures we found that lie near the boundary of $GU_f$ have quite complicated multiscale architectures and thus would be difficult to find numerically. Also, it is not clear whether there are signifiicanty simpler microstructures that can do the job. The numerical route of Sigmind, should provide some simpler alternatives 
for the strut configurations in the multimode structures in the walls, although even then one needs to make subtle multiscale replacements (such as those appearing later in \fig{7} and \fig{8}) to achieve the desired performance.
Numerical tests need to be made to see whether or not one can achieve the same performance with simpler structures.
While strut configurations might be suitable at low volume fractions they are unlikely to be ideal at high volume fractions.
Work by Allaire and Aubry \cite{Allaire:1999:OMP} shows that sometimes optimal microstructures necessarily have structure on multiple length scales. Even if one could numerically explore the question, it is not clear how one could summarize the results in any concise way.

From the applied side there is growing interest in trying to characterize the effective elasticity tensors of microstructures that can be produced by three-dimensional or two-dimensional printing. A dramatic example of such a microstructure is given in \fig{0}. Our results have obvious relevance to this problem 
in the case where the $3-d$ printed material uses only one isotropic material plus void. Although our microstructures are somewhat extreme, they provide benchmarks that show what is theoretically possible. What is possible in practice will be a subset of this. 

The microstructures we consider involve taking three limits. First, as they have structure on multiple length scales, the homogenization limit where the ratio between
length scales goes to infinity needs to be taken. Second, the limit $\Gd\to 0$ needs to be taken. Third, as the structure involves this walls of width $\Ge$, along which the material can "slip", the limit $\Ge\to 0$ needs to be taken so the contribution to the complementary energy of these walls goes to zero, when the structure supports an applied stress. 
(Here $\Ge$ should not be confused with the size of the unit cell, as common in homogenization theory).
The limits should be taken in this order, as, for example, standard homogenization theory is justified
only if $\Gd\ne 0$, so we need to take the homogenization limit before taking the limit $\Gd\to 0$. In the walled 
structures the material may only occupy a small volume fraction, but this is ulimately irrelevant as the thin walled
structures themselves only occupy only a very small volume fraction in the final material (that goes to zero as
$\Ge\to 0$).

The case, applicable to printed materials, when phase 2 is actually void, rather than almost void, requires special care.
To justify the homogenization steps taken here one has to first replace the void phase 2 with a composite foam having a small amount of phase1 as the matrix phase, so that its effective elastcity tensor is nonzero, but approaches zero as the proportion of phase 1 in it tends to zero. The microgeometry in this composite needs to be much smaller than the scales in the geometries discussed here, which would involve mixtures of it and phase 1.

We emphasize, too, that our analysis is valid only for linear elasticity, and ignores nonlinear effects  such as buckling.
In reality the structures will easily buckle under compression. This buckling will occur, for example, in the walled pentamode structure of \fig{8}. Additionally, some of the multimode materials are constructed via a superposition of 
appropriately shifted and deformed pentamode materials, and these substructures will interact under finite deformations.
Also, in practice it would be difficult to realize the delicate multiscale materials that come close to attaining the bounds.
Thus what is practically realizable will just be a subset, dependent of the current state of technology, of the set
 $GU_f$.

We also remark that while the title refers only to printed materials, the results are also applicable to 
any periodic, or statistically homogeneous, material containing voids or pores in a homogeneous material. 
Printed materials are more interesting than typical porous materials as they allow one to explore a wider range of 
interesting structures. 

In a companion paper \cite{Milton:2016:TCC} we consider the opposite limit $\Gd\to \infty$ corresponding to a mixture of an elastic material and an almost rigid material.

\section{A review of some bounds on the elastic moduli of two-phase composites and geometries that attain them.}
\labsect{review}
\setcounter{equation}{0}
Here we review a selection of results on sharp bounds on the elastic response of two-phase composites and the associated problem of identifying
optimal geometries that attain them. The interested reader is also encouraged to look at the books of  
\cite{Nemat-Nasser:1999:MOP}, \cite{Cherkaev:2000:VMS}, \cite{Milton:2002:TOC}, \cite{Allaire:2002:SOH}, \cite{Torquato:2002:RHM}, \cite{Tartar:2009:GTH}
which provide a much more comprehensive survey.

The most elementary bounds on the elastic properties of composites, are the classical bounds of Hill 
\cite{Hill:1952:EBC}
who implicitly showed that
\beq \lang[\BC(\Bx)]^{-1}\rang^{-1} \leq \BC_* \leq \lang\BC(\Bx)\rang.
\eeq{0.1}
Here the angular brackets $\lang\cdot\rang$ denote a volume average, and the inequality holds in the sense 
of quadratic forms, i.e., for fourth order tensors $\BA$ and $\BB$ satisfying the symmetries of elasticity
tensors we say that $\BA\geq\BB$ if $\BGe:\BA\BGe\geq\BGe:\BB\BGe$ for all matrices $\BGe$. While
these bounds were not explicitly stated by Hill in his 1952 paper they are an immediate and obvious consequence
of his equation (2). 
If the two phases are isotropic 
the spectral decomposition of the elasticity tensors $\BC_1$ and $\BC_2$ of the two phases is
\beq \BC_1=3\Gk_1\BGL_h+2\Gm_1\BGL_s~~{\rm and}~~
\BC_2=3\Gk_2\BGL_h+2\Gm_2\BGL_2 \eeq{0.3}
where $\Gk_1$ and $\Gk_2$ are the bulk-moduli of the two phases, and $\Gm_1$ and $\Gm_2$ are the bulk-moduli,
and 
\beq \{\BGL_h\}_{ijk\ell}={\textstyle {1\over3}}\Gd_{ij}\Gd_{k\ell},~~~
\{\BGL_s\}_{ijk\ell}={\textstyle {1\over2}}[\Gd_{ik}\Gd_{j\ell}+\Gd_{i\ell}\Gd_{kj}]
-{\textstyle {1\over3}}\Gd_{ij}\Gd_{k\ell} \eeq{0.2}
act as projections. The tensor $\BGL_h$ projects onto
the one-dimensional space of matrices proportional
to the second-order identity matrix, while $\BGL_s$ projects onto
the five-dimensional space of trace free matrices. 
Similarly if the effective elasticity tensor $\BC_*$ is isotropic we have that $\BC_*=3\Gk_1\BGL_h+2\Gm_1\BGL_s$,
where $\Gk_*$ and $\Gm_*$ are the effective bulk and shear moduli of the composite. 
In this paper we are interested in the case where the two phases are well-ordered in the sense that
\beq \BC_1\geq\BC_2, \eeq{0.3a}
and we will take the limit as $\BC_2\to 0$, meaning that all the eigenvalues of $\BC_2$ approach zero.
In the case of isotropic components this well-ordering assumption is satisfied if  $\Gk_1 \geq \Gk_2$ and $\Gm_1 \geq \Gm_2$, and we will take the limit as both $\Gk_2$ and $\Gm_2$ approach zero.

For isotropic composites of two well-ordered materials Hashin and Shtrikman \cite{Hashin:1963:VAT} and Hill \cite{Hill:1963:EPR}
obtained the celebrated bounds 
\beqa
 \Gk_* & \geq & f\Gk_1+(1-f)\Gk_2
-\frac{f(1-f)(\Gk_1-\Gk_2)^2}{(1-f)\Gk_1+ f\Gk_2+4\Gm_2/3}, \nonum
\Gk_* & \leq & f\Gk_1+(1-f)\Gk_2
-\frac{f(1-f)(\Gk_1-\Gk_2)^2}{(1-f)\Gk_1+ f\Gk_2+4\Gm_1/3}, \nonum
\Gm_* & \geq &  f\Gm_1+(1-f)\Gm_2
-\frac{f(1-f)(\Gm_1-\Gm_2)^2}{(1-f)\Gm_1+ f\Gm_2+\Gm_2(9\Gk_2+8\Gm_2)/[6(\Gk_2+2\Gm_2)]},\nonum
\Gm_* & \leq &  f\Gm_1+(1-f)\Gm_2
-\frac{f(1-f)(\Gm_1-\Gm_2)^2}{(1-f)\Gm_1+ f\Gm_2+\Gm_1(9\Gk_1+8\Gm_1)/[6(\Gk_1+2\Gm_1)]}.
\eeqa{0.4}
In fact these bounds (and the variational principles they derive from) hold even if one component has a negative 
bulk modulus, so long as the composite is stable \cite{Kochmann:2015:RBE}.
For two-dimensional
composites (fiber reinforced materials) analogous
bounds on the effective elastic moduli were found by
Hill \cite{Hill:1964:TMP} and Hashin \cite{Hashin:1965:EBF}. Bounds on the complex effective bulk and shear moduli of isotropic two-phase, two or three-dimensional
composites were also obtained \cite{Gibiansky:1993:EVM,Milton:1997:EVM,Gibiansky:1999:EVM,Gibiansky:1993:BCB,Gibiansky:1993:VCE,Gibiansky:1997:BCB}: these are appropriate to the propagation of fixed frequency elastic waves in composites when
one or both of the phases is viscoelastic, and when the wavelength is much larger than the microstructure. 

An important ``attainability principle'' is that bounds obtained by substituting a trial field in a variational principle
will be attained when the geometry is such that the actual field matches this trial field. This
principle was used, for example, in \cite{Milton:1981:CBT} to find geometries that attain the Hashin-Shtrikman bounds on the
effective bulk modulus of composites with three or more phases (see also \cite{Gibiansky:2000:MCE}). The Hashin-Shtrikman 
variational principles involve a minimization over trial polarization fields, and the actual polarization field
depends on the choice of the elasticity tensor $\BC_0$ of a ``reference medium'' (typically chosen to be positive definite) and is defined by
\beq \BP(\Bx)=(\BC(\Bx)-\BC_0)\BGe(\Bx)=\BGs(\Bx)-\BC_0\BGe(\Bx). \eeq{0.5}
The variational principles require that $\BC(\Bx)-\BC_0$ be either positive semidefinite or negative semidefinite,
so in the case of a well-ordered material natural choices of $\BC_0$ are $\BC_1$ or $\BC_2$ and correspondingly 
the field will be zero respectively in phase 1 or phase 2. The bounds are obtained by assuming it is constant in the other phase
(proportional to the identity in case of the bulk modulus bounds, and trace free for the shear modulus bounds). 
Hashin and Shtrikman recognized \cite{Hashin:1963:VAT} that the effective bulk modulus would be attained by the Hashin assemblage of coated spheres \cite{Hashin:1962:EMH}. A single coated
sphere can be a neutral inclusion: if the surrounding ``matrix'' material has an appropriate bulk modulus (with a specific value between $\Gk_1$ and $\Gk_2$)
one can insert it in the ``matrix'' material without disturbing a surrounding hydrostatic field (this is the principle behind the unfeelability 
cloak of B{\"u}ckmann, Thiel, Kadic, Schittny and Wegener \cite{Buckmann:2014:EMU}). The inclusion is invisible to the surrounding field and one can 
continue to insert similar inclusions, scaled to sizes ranging to the very small,
until one essentially obtains a two--phase composite with effective bulk modulus the same as the original ``matrix'' material. Due to radial symmetry
the forces acting on the spherical inner core will be equally distributed around the boundary and directed radially: thus the field inside the core material
is hydrostatic and constant, and hence by the ``attainability principle'', and due to their neutrality, sphere assemblages must attain the
effective bulk modulus bounds in \eq{0.4}.

One very important class of microgeometries for which the field is constant in one phase are the 
sequentially layered laminates (first introduced by Maxwell \cite{Maxwell:1954:TEM}) built by layering phase 2 with phase 1 in a direction $\Bn_1$ (by which we mean $\Bn_1$ is perpendicular to the layers), 
then taking this laminate and layering it again on a much larger length scale with phase 1 in a direction $\Bn_2$ to obtain a ``rank 2'' laminate, and continuing 
this process until one obtains a rank $m$-laminate, containing in a sense a ``core'' of phase 2 surrounded by layers of phase 1. The field is then constant in the core material
of phase 2. An explicit formula for the
effective elasticity tensor of such sequentially layered laminates was obtained by Francfort and Murat \cite{Francfort:1986:HOB} generalizing the analogous formulas obtained by 
Tartar \cite{Tartar:1985:EFC} for conductivity.
Of course one can switch the roles of the phases in this construction and thus obtain a material where the field is constant in phase 1. It then immediately follows
from the ``attainability principle''(without requiring any calculation!) that one can attain the Hashin-Shtrikman shear modulus bounds \eq{0.4} 
(and simultaneously the bulk modulus bounds)
if one can find a sequentially layered laminate
that has an isotropic elasticity tensor, and the easiest way to do this is to do the lamination sequentially by adding infinitesimal layers in random directions. This established the
attainability of the Hashin-Shtrikman shear modulus bound \cite{Milton:1986:MPC}, also established independently and at the same time by Norris \cite{Norris:1985:DSE},
using the differential scheme that was known to be realizable \cite{Milton:1985:SEM, Avellaneda:1987:IHD}--in fact Roscoe \cite{Roscoe:1973:ICE} had earlier realized the differential
approximation scheme could produce the desired shear modulus--and at the same time elegantly by Francfort and Murat \cite{Francfort:1986:HOB}, using sequentially layered laminates with 
just five directions of lamination (in the case of three dimensional composites). 

Hill \cite{Hill:1963:EPR} proved that the bulk modulus bounds are valid also in the non-well-ordered case
where $\Gm_1\geq\Gm_2$ but $\Gk_1\leq\Gk_2$.
As far as we are aware, the tightest bounds on the effective shear modulus of three-dimensional composites 
in the non-well-ordered case where $\Gm_1\geq\Gm_2$ but $\Gk_1\leq\Gk_2$
are those of Milton and Phan-Thien \cite{Milton:1982:NBE}
\beqa
&~&\min_{\substack{\Gz \\ 0 \leq \Gz \leq 1}}
\frac{8\lang 6/\Gm+7/\Gk\rang_{\Gz}+15/\Gm_2}
{2(\lang 21/\Gm+ 2/\Gk \rang_{\Gz}/\Gm_2
+ 40\lang 1/\Gm\rang_{\Gz}\lang 1/\Gk\rang_{\Gz})} \nonum
&~&\quad\quad\quad\quad\leq \frac{f(1-f)(\Gm_1-\Gm_2)^2}
{f\Gm_1+ (1-f)\Gm_2-\Gm_*}-(1-f)\Gm_1-f\Gm_2
\nonum
&~&\quad\quad\quad\quad\quad\quad\quad\quad\leq \max_{\substack{\Gz \cr \\ 0 \leq \Gz \leq 1 \cr}}
\frac{8 \Gm_1\lang 6\Gk+7\Gm\rang_{\Gz}
+15\lang\Gm\rang_{\Gz}\lang\Gk\rang_{\Gz}}
{2(\lang 21\Gk+2\Gm\rang_{\Gz}+40\Gm_1)}, \eeqa{0.6}
where for any quantity $a$ taking values $a_1$ and $a_2$ in phase $1$ and phase $2$, respectively, we define
$\lang a \rang_{\Gz}\equiv \Gz a_1+(1-\Gz)a_2$.
These bounds are obtained by eliminating the geometric parameters from the bounds of Milton and Phan-Thien \cite{Milton:1982:NBE}
and are tighter than the more well-known Walpole bounds \cite{Walpole:1966:BOE}, and in fact sharp (as they coincide with the Hashin-Shtrikman formula,
which correspond to particular geometries as we have discussed)
when the moduli are slightly non-well-ordered. Specifically, the first bound in \eq{0.6} is sharp when the minimum over $\Gz$ is attained at $\Gz=0$, which occurs when
\beq \Gk_1-\Gk_2 \geq -\frac{(3\Gk_2+8\Gm_2)^2}{42\Gk_2^2}\frac{\Gk_1\Gk_2}{\Gm_1\Gm_2}(\Gm_1-\Gm_2)
\eeq{0.7}
and the second bound in \eq{0.6} is sharp when the maximum over $\Gz$ is attained at $\Gz=1$, which occurs when
\beq \Gk_1-\Gk_2 \geq -\frac{(3\Gk_1+8\Gm_1)^2}{42\Gm_1^2}(\Gm_1-\Gm_2).
\eeq{0.8}
The bounds \eq{0.4} and \eq{0.6} constrain the pair $(\Gk_*,\Gm_*)$ 
to lie in a rectangular box. Berryman and Milton \cite{Berryman:1988:MRC} 
obtained tighter coupled bounds which slice off two opposing corner regions of the box by eliminating the geometric parameters from the 
bulk modulus bounds of Beran \cite{Beran:1966:UCV} (as simplified by Milton \cite{Milton:1981:BEE}) and from the
shear modulus bounds of Milton and Phan-Thien \cite{Milton:1982:NBE}. There is good reason to believe these bounds can be improved as the 
analogous two-dimensional bounds are not as tight as the
bounds of Cherkaev and Gibiansky \cite{Cherkaev:1993:CEB} coupling $\Gk_*$ and $\Gm_*$ which were derived using the translation method.

For anisotropic composites with an effective tensor $\BC_*$, the microstructure independent bounds that are directly analogous
to the Hashin-Shtrikman-Hill bounds, given by \eq{0.4}, are the ``Trace bounds'' 
\beqa  
f\Tr[\BGL_h(\BC_*-\BC_2)^{-1}]
& \leq & \frac{1}{3(\Gk_1-\Gk_2)}+\frac{(1-f)}{3\Gk_2+4\Gm_2},\nonum
(1-f)\Tr[\BGL_h(\BC_1-\BC_*)^{-1}]
& \leq & \frac{1}{3(\Gk_1-\Gk_2)}-\frac{f}{3\Gk_1+4\Gm_1},\nonum
f\Tr[\BGL_s(\BC_*-\BC_2)^{-1}] 
&\leq & \frac{5}{2(\Gm_1-\Gm_2)}+\frac{3(\Gk_2+2\Gm_2)(1-f)}{\Gm_2(3\Gk_2+4\Gm_2)},\nonum
 (1-f)\Tr[\BGL_s(\BC_1-\BC_*)^{-1}] 
&\leq & \frac{5}{2(\Gm_1-\Gm_2)}-\frac{3(\Gk_1+2\Gm_1)f}{\Gm_1(3\Gk_1+4\Gm_1)},
\eeqa{0.9}
obtained independently by Milton and Kohn \cite{Milton:1988:VBE} and Zhikov \cite{Zhikov:1988:ETA,Zhikov:1991:ETA,Zhikov:1991:EHM}. In these
expressions the fourth order tensors $\BGL_h$ multiply the fourth order tensors on their right, and
\beq \Tr[\BA]=A_{ijij} \eeq{0.10}
defines the ``Trace'' of a fourth order tensor (see also Francfort and Murat \cite{Francfort:1986:HOB} and Nemat-Nasser and Hori \cite{Nemat-Nasser:1993:MOP} for related bounds). 
From the ``attainability principle'' it follows that these bounds will be achieved whenever the composite 
is a sequentially layered laminate, with a ``core'' of one phase, surrounded by layers (on widely separated length scales) of the other phase. 
When $\BC_*$ is isotropic these bounds \eq{0.9} reduce to the Hashin-Shtrikman-Hill bounds \eq{0.4}. In the case the two phases, and hence the composite, are
incompressible we can define the five effective shear moduli $\Gm_1^*, \Gm_2^*, \Gm_3^*, \Gm_4^*, \Gm_5^*$ to be the five finite eigenvalues of $\BC_*/2$,
and the second pair of bounds in \eq{0.9} reduce to  
\beqa
\sum_{i=1}^5\frac{f}{2(\Gm_{*i}-\Gm_2)} & \leq &
\frac{5}{2(\Gm_1-\Gm_2)}+\frac{3(\Gk_2+2\Gm_2)(1-f)}{\Gm_2(3\Gk_2+4\Gm_2)}, \nonum
\sum_{i=1}^5\frac{(1-f)}{2(\Gm_1-\Gm_{*i})} & \leq &
\frac{5}{2(\Gm_1-\Gm_2)}-\frac{3(\Gk_1+2\Gm_1)f}{\Gm_1(3\Gk_1+4\Gm_1)}.
\eeqa{0.11}
Lipton \cite{Lipton:1988:EET} established that the analogous bounds for the two effective 
shear moduli $\Gm_1^*$ and $\Gm_2^*$ of two-dimensional composites of two incompressible isotropic phases
completely characterize $GU_f$.

Earlier, Willis \cite{Willis:1977:BSC} had considered anisotropic composites and
used the Hashin-Shtrikman variational principle with a trial polarization that was zero in one phase and constant in the other to
obtain bounds on the elastic energy of a two-phase composite. He found that these bounds are not microgeometry independent, but rather
involve the two-point correlation function, i.e., the probability that a rod with fixed orientation lands with both ends in phase $1$ 
when thrown randomly in a composite. It follows from the ``attainability principle'' that the Willis bounds will be achieved when the composite 
is a sequentially layered laminate, with a ``core'' of one phase, surrounded by layers (on widely separated length scales) of the other phase. 

In a major advance, Avellaneda \cite{Avellaneda:1987:OBM} recognized that for any composite of two phases with well-ordered tensors 
not all the information contained in the two-point correlation function
was relevant to determining the bounds: what was relevant was a ``reduced two-point correlation function'' that could be represented as a positive
measure $\mu(\BGx)$ (with unit integral) on the sphere $|\BGx|$=1.
Roughly speaking one takes the Fourier transform of the two-point correlation function and integrates it over
rays $\Bk=k\BGx$ in ``Fourier space'' keeping $\BGx$ fixed and integrating over $k$ from $0$ to infinity. Most importantly, every such a measure could be realized to an 
arbitrarily high degree of approximation by the measure of a suitable sequentially layered laminate. For example, a measure with weighted delta functions in directions 
$\BGx_1$ and $\BGx_2$ would be realized by a second rank sequentially layered laminate with layers normal to $\BGx_1$ and $\BGx_2$. (We note in passing that these 
``reduced two-point correlation functions''
of Avellaneda, are a special case of the $H$-measures introduced at the same time by Tartar \cite{Tartar:1989:MSA,Tartar:1990:MNA}, in terms of which he
could calculate second order corrections to the effective tensor of a nearly homogeneous composite. $H$-measures were also introduced independently by G{\'e}rard
\cite{Gerard:1989:CPC,Gerard:1994:MAC}
under the name of microlocal defect measures. For composites of two isotropic phases the Hashin-Shtrikman conductivity bounds, and indeed variational 
conductivity bounds at any order,
can be naturally expressed in terms of the series expansion coefficients of the effective tensor up to a corresponding order for a nearly homogeneous composite as shown by 
Milton and McPhedran \cite{Milton:1982:CTM}).

The fantastic implication was that 
by summing the Willis bounds \cite{Willis:1977:BSC}, and then minimizing over all positive measures on the sphere, one would get sharp bounds on the sum of elastic complementary energies
\beq W_f^0(\BGs^0_1,\BGs^0_2,\BGs^0_3,\BGs^0_4,\BGs^0_5,\BGs^0_6)
 = \min_{\BC_*\in GU_f}\sum_{j=1}^6\BGs^0_j:\BC_*^{-1}\BGs^0_j, \eeq{0.12}
and similarly one could get sharp bounds on the sum of elastic energies
\beq
 W_f^6(\BGe^0_1,\BGe^0_2,\BGe^0_3,\BGe^0_4,\BGe^0_5,\BGe^0_6)
 =  \min_{\BC_*\in GU_f}\sum_{i=1}^6\BGe^0_i:\BC_*\BGe^0_i.
\eeq{0.13}
Here some of the applied
stresses $\BGs_j^0$ or applied strains $\BGe^0_i$ could be zero. 
Thus the evaluation of the functions 
$W_f^0(\BGs^0_1,\BGs^0_2,\BGs^0_3,\BGs^0_4,\BGs^0_5,\BGs^0_6)$ and  $W_f^6(\BGe^0_1,\BGe^0_2,\BGe^0_3,\BGe^0_4,\BGe^0_5,\BGe^0_6)$
reduces to a finite dimensional minimization problem which can be done numerically. Hence
we will treat the functions $W_f^0(\BGs^0_1,\BGs^0_2,\BGs^0_3,\BGs^0_4,\BGs^0_5,\BGs^0_6)$ as being known, and we will call an ``Avellaneda material'' an
associated sequentially layered laminate material with effective tensor $\BC_*=\BC_f^A(\BGe^0_1,\BGe^0_2,\BGe^0_3,\BGe^0_4,\BGe^0_5,\BGe^0_6)$ 
that attains the minimum in \eq{0.13}, and similarly we call a ``complementary Avellaneda material'' an 
associated sequentially layered laminate material with effective tensor $\BC_*=\widetilde{\BC}_f^A(\BGs^0_1,\BGs^0_2,\BGs^0_3,\BGs^0_4,\BGs^0_5,\BGs^0_6)\in GU_f$
that attains the minimum in \eq{0.12}. Explicit analytical formulas for the tensors $\BC_f^A(\BGe^0_1,\BGe^0_2,\BGe^0_3,\BGe^0_4,\BGe^0_5,\BGe^0_6)$ and
$\widetilde{\BC}_f^A(\BGs^0_1,\BGs^0_2,\BGs^0_3,\BGs^0_4,\BGs^0_5,\BGs^0_6)$ are not generally available, but rather have to be found by numerical computation. When $\BC_1\geq\BC_2$ one needs to take the minimum in \eq{0.12} over the $\BC_*$ of sequentially layered laminates with a ``core material'' of phase 2. Similarly, when $\BC_1^{-1}\geq\BC_2^{-1}$ the minimum in \eq{0.13} also can be taken over the $\BC_*$ of sequentially layered laminates with a ``core material'' of phase 2. We remark that although Avellaneda assumed the tensors
$\BC_1$ and $\BC_2$ were isotropic, his analysis easily extends to the case where the tensors are anisotropic
but well-ordered (either with $\BC_1\geq\BC_2$ or $\BC_2\geq\BC_1$) and with constant orientation throughout the composite: see, for example, Section 23.3 in \cite{Milton:2002:TOC}. 

These $\BC_*$ of sequentially layered laminates are
given by the formula of Francfort and Murat \cite{Francfort:1986:HOB} and Gibiansky and Cherkaev \cite{Gibiansky:1987:MCE}:
\beq (1-f)(\BC_1-\BC_*)^{-1}=(\BC_1-\BC_2)^{-1}-f\sum_{j=1}^r c_j\BGG(\Bn_j), \eeq{0.13.0}
where $r$ is the rank of the sequential laminate, the positive weights $c_j$ sum to $1$, the $\Bn_i$ are the lamination directions, and $\BGG(\Bn)$ is the fourth order tensor with
elements given by 
\beq
\begin{array}{l}
\{\BGG(\Bn)\}_{hik\ell}=
\\*[.2em]
\displaystyle \frac{1}{4}\left(n_h\{\BC(\Bn)^{-1}\}_{ik}n_\ell+n_h\{\BC(\Bn)^{-1}\}_{i\ell}n_k+n_i\{\BC(\Bn)^{-1}\}_{hk}n_\ell+n_i\{\BC(\Bn)^{-1}\}_{h\ell}n_k\right),
\end{array}
\eeq{0.13.1}
in which $\BC(\Bn)=\Bn\cdot\BC_1\Bn$ is the $3\times 3$ matrix, known as the acoustic tensor, with elements 
\beq \{\BC(\Bn)\}_{ik}=\{\Bn\cdot\BC_1\Bn\}_{ik}=n_h\{\BC_1\}_{hik\ell}n_\ell. \eeq{0.13.2}
Thus the minimum needs to be taken over the rank $r$ of the sequential laminate, over the positive weights $c_j$, that sum to $1$, and over the lamination directions $\Bn_j$.
In the case phase $1$ is isotropic, with bulk modulus $\Gk_1$ and shear modulus $\Gm_1$, $\BC(\Bn)$ can be easily calculated and one obtains
\beq \{\BGG(\Bn_j)\}_{hik\ell}=\frac{3n_h n_i n_k n_\ell}{3\Gk_1+4\Gm_1}+\frac{1}{4\Gm_1}\left(n_h\Gd_{ik}n_\ell+n_h\Gd_{i\ell}n_k+n_i\Gd_{hk}n_\ell+n_i\Gd_{h\ell}n_k-4n_h n_i n_k n_\ell\right).
\eeq{0.13.3}
Francfort, Murat, and Tartar  \cite{Francfort:1995:FOM} prove that when $\BC_1$ is isotropic it suffices to limit attention to laminates of rank $r\leq 6$. When
$\BC_1$ is anisotropic we extend an argument due to Avellaneda \cite{Avellaneda:1987:OBM}. Consider the set $\CA$ consisting of all fourth-order tensors $\BA$ of the 
form
\beq \BA=\int_{|\Bn|=1}\BGG(\Bn)\,m(d\Bn), \eeq{0.13.3a}
where $m(d\Bn)$ is a non-negative measure on the unit sphere having an integral of $1$ over the sphere. Since $\BA$ satisfies
\beq  \{\BA\}_{hik\ell}\{\BC_1\}_{hik\ell}=\int_{|\Bn|=1}\{\BC(\Bn)^{-1}\}_{ik}\{\BC(\Bn)\}_{ik}\,m(d\Bn)=3 \eeq{0.13.3b}
it follows that $\CA$ is a convex set in a space of dimension $\Gv=20$ (with 20 of the 21 independent matrix elements of $\BA$ as coordinates, and the remaining
element being determined by \eq{0.13.3b}). The extreme points correspond to point masses on the unit sphere. Hence any tensor of the form \eq{0.13.3a} is
a convex combination of  at most $\Gv+1$ extreme points. Thus the sum \eq{0.13.0} can be limited to $r\leq 21$, i.e., it suffices to consider laminates up to rank $21$.
Lipton \cite{Lipton:1991:BEC,Lipton:1992:BPS,Lipton:1994:OBE}
obtained a complete algebraic characterization of the possible 
sequentially layered laminates having transverse or orthotropic symmetry and derived explicit expressions for many of the associated bounds. 
The Avellaneda materials are of course difficult to build in practice since they have structure on multiple length
scales. However if $f$ is small and one phases is void, Bourdin and Kohn \cite{Bourdin:2008:OST} show that it suffices to use a walled structure (similar
to the structure in \fig{2}(b), but with walls in many directions, not just two, and with the wall thickness depending on their orientation). 

As observed by Avellaneda \cite{Avellaneda:1987:OBM}, the implications of course also apply to two-dimensional elasticity. Defining
\beq W_f^0(\BGs^0_1,\BGs^0_2,\BGs^0_3)
 = \min_{\BC_*\in GU_f}\sum_{j=1}^3\BGs^0_j:\BC_*^{-1}\BGs^0_j, \eeq{0.13a}
and 
\beq
 W_f^3(\BGe^0_1,\BGe^0_2,\BGe^0_3)
 =  \min_{\BC_*\in GU_f}\sum_{i=1}^3\BGe^0_i:\BC_*\BGe^0_i,
\eeq{0.13b}
then there is an Avellaneda material with effective tensor $\BC_*=\BC_f^A(\BGe^0_1,\BGe^0_2,\BGe^0_3)$ that attains the minimum in \eq{0.13b},
and a complementary Avellaneda material with effective tensor $\BC_*=\widetilde{\BC}_f^A(\BGs^0_1,\BGs^0_2,\BGs^0_3)\in GU_f$ that attains the minimum in \eq{0.13a}. In two-dimensional elasticity, sequentially 
layered laminates have elasticity tensors given by \eq{0.13.0} -\eq{0.13.2} when the tensor $\BC_1$ is anisotropic. When the elasticity tensor $\BC_1$ of phase 1 is isotropic, the sequentially
layered laminates of rank $r$ have effective compliance tensors $\BS_*=(\BC_*)^{-1}$ given by 
the Gibiansky-Cherkaev formula \cite{Gibiansky:1987:MCE}
\beq (1-f)(\BS_1-\BS_*)^{-1}=(\BS_1-\BS_2)^{-1}-f[(4\Gk_2)^{-1}+(4\Gm_2)^{-1}]\BM, \eeq{0.13c}
[see their equations (2.37) and (2.38) and see also Lurie, Cherkaev, and Fedorov \cite{Lurie:1982:RODb}
who derive an equivalent, but less simple, formula], where  $\BS_1=(\BC_1)^{-1}$ and $\BS_2=(\BC_2)^{-1}$ are the compliance tensors of the two phases, occupying respectively volume fractions $f$ and $1-f$,
and $\BM$ has elements
\beq \{\BM\}_{hik\ell}=\sum_{j=1}^r c_j{\Bt_j}_h{\Bt_j}_i{\Bt_j}_k{\Bt_j}_\ell, \eeq{0.13ca}
in which the $\Bt_j$ are unit vectors perpendicular to the directions of lamination (i.e., parallel to the layer boundaries), and the $c_j$ are any set of positive weights, summing to 1,
giving the proportions of phase 1 laminated in the various directions.  The tensor $\BM$ is clearly positive semidefinite and has the properties that
\beq  \{\BM\}_{hkhk}= \{\BM\}_{hhkk}=1. \eeq{0.13d}
Conversely, Avellaneda and Milton \cite{Avellaneda:1989:BEE} have shown that given positive semidefinite fourth order tensor $\BM$ satisfying \eq{0.13d} there is a sequential layered laminate 
of rank $r\leq 3$ that corresponds to it,
i.e., such that \eq{0.13c} holds for some choice of unit vectors $\Bt_j$ and weights $c_j$ (see also Theorem 2.2 of  \cite{Francfort:1995:FOM}).
Thus when $\BC_1$ is isotropic, the computation of the complementary Avellaneda tensor $\widetilde{\BC}_f^A(\BGs^0_1,\BGs^0_2,\BGs^0_3)$
reduces to a minimization over positive semidefinite fourth order tensors $\BM$ satisfying \eq{0.13d}. When $\BC_1$ is anisotropic, by the same argument as in the three dimensional case, it suffices to consider sequential layered  laminates of rank at most $6$.

We also remark that aside from hierarchical laminates there are many other structures that have a uniform field in one phase, sometimes only for
certain applied fields. These include assemblages of confocal ellipses and ellipsoids \cite{Milton:1980:BCD, Milton:1981:BCP, Grabovsky:1995:MMEa}, the periodic Vigdergauz geometries
\cite{Vigdergauz:1986:EEP, Vigdergauz:1994:TDG, Grabovsky:1995:MMEb, Vigdergauz:1996:RLE, Vigdergauz:1999:EMI}, the Sigmund structures \cite{Sigmund:2000:NCE}, and the
periodic $E$-inclusions of Liu, James, and Leo \cite{Liu:2007:PIM} (see also section 23.9 of \cite{Milton:2002:TOC}). Usually these attain the bounds when the measure $\mu(\BGx)$
minimizing the sum of Willis bounds is not required to be a discrete measure. Allaire and Aubry \cite{Allaire:1999:OMP} have shown that sometimes the best microstructure
necessarily has structure on multiple length scales (like sequentially layered laminates).

For single energies for anisotropic two-phase composites, the Hill bounds \eq{0.1} imply
\beqa
\BGe_0:[f\BC_1^{-1}+(1-f)\BC_2^{-1}]^{-1}\BGe_0 & \leq & \BGe_0:\BC_*\BGe_0 \leq \BGe_0:[f\BC_1+(1-f)\BC_2]\BGe_0, \nonum
\BGs_0:[f\BC_1+(1-f)\BC_2]^{-1}\BGs_0 &  \leq &\BGs_0:\BC_*^{-1}\BGs_0 \leq \BGs_0:[f\BC_1^{-1}+(1-f)\BC_2^{-1}]\BGs_0.\nonum
&~ &
\eeqa{0.14}
Improved, and in fact sharp, upper and lower bounds on the elastic energy $\BGe_0:\BC_*\BGe_0$ in terms of the given 
applied strain $\BGe_0$, and 
sharp upper and lower bounds on the complimentary elastic energy $\BGs_0:\BC_*^{-1}\BGs_0$ in terms of the  given applied stress $\BGs_0$ were obtained 
for isotropic component materials by Gibiansky and Cherkaev \cite{Gibiansky:1984:DCPa}, Kohn and Lipton \cite{Kohn:1988:OBE}, and
Allaire and Kohn \cite{Allaire:1993:EOB, Allaire:1993:OBE, Allaire:1994:OLB}. The paper of Gibiansky and Cherkaev 
\cite{Gibiansky:1984:DCPa} was for the fourth-order plate equation, but this can mapped to the equivalent two-dimensional 
elasticity problem considered by Allaire and Kohn \cite{Allaire:1993:EOB}. Their lower bounds on $\BGs^0:\BC_*^{-1}\BGs^0$
are equivalent to the bounds that for any 
tensor $\BC_*\in GU_f$,
\beq \BGs^0:\BC_*^{-1}\BGs^0\geq \BGs^0:[\widetilde{\BC}_f^A(\BGs^0,0,0)]^{-1}\BGs^0, \eeq{0.15}
and they provided an explicit formula for the right hand side for any $2\time 2$ symmetric matrix $\BGs^0$ representing the applied stress. This bound can be viewed
in two ways: in the way originally interpreted, i.e., as a bound on the possible (elastic energy, average stress, volume fraction) triplets;
or as a bound 
\beq \BGs^0:\BGe^0\geq \BGs^0:[\widetilde{\BC}_f^A(\BGs^0,0,0)]^{-1}\BGs^0, \eeq{0.16}
on the possible (average stress, average strain, volume fraction) triplets. Here $\BGe^0=\BC_*^{-1}\BGs^0$ is the strain associated with
$\BGs^0$. Significantly, Milton, Serkov, and Movchan \cite{Milton:2003:RAS} found
that the inequality $\eq{0.16}$ completely characterizes the possible (average stress, average strain, volume fraction) triplets in the limit
in which one phase becomes void, when the other phase is isotropic. Specifically, given any triplet $(\BGs^0,\BGe^0,f)$ satisfying
\eq{0.16} as an inequality, they give a recipe for constructing a two-dimensional microstructure with effective tensor $\BC_*$ and 
having phase 1 occupy a volume fraction $f$ such that $\BGs^0=\BC_*\BGe^0$. 

For three dimensional composites explicit expressions for the optimal upper energy bound have been found 
by Gibiansky and Cherkaev  \cite{Gibiansky:1987:MCE} and Allaire \cite{Allaire:1994:ELP} for the
case of a two-phase composite where one of the phases is void or rigid \cite{Gibiansky:1987:MCE}.
Grabovsky \cite{Grabovsky:1996:BEM} obtained energy bounds for two-phase composites containing anisotropic phases, each with a constant orientation.

Another major advance was made by Milton and Cherkaev \cite{Milton:1995:WET} 
who showed that any desired positive definite fourth order tensor which has the symmetries of an elasticity tensor
could be realized as the effective elasticity tensor $\BC_*$ of a composite of a sufficiently stiff isotropic material and a sufficiently compliant isotropic material. One key
to this advance was the realization that certain structures called pentamode materials could be (arbitrarily) stiff to one applied stress $\BGs^0_1$ and yet have five mutually
orthogonal strains $\BGe^0_1$, $\BGe^0_2$, $\BGe^0_3$, $\BGe^0_4$, $\BGe^0_5$, each orthogonal to $\BGs^0_1$ as five (arbitrarily compliant) easy modes of deformations (hence the name pentamode). For such a pentamode $W_f^5(\BGs^0_1,\BGe^0_1,\BGe^0_2,\BGe^0_3,\BGe^0_4,\BGe^0_5)$ 
\beq W_f^5(\BGs^0_1,\BGe^0_1,\BGe^0_2,\BGe^0_3,\BGe^0_4,\BGe^0_5)
 =  \min_{\BC_*\in GU_f}\left[\left(\sum_{i=1}^5\BGe^0_i:\BC_*\BGe^0_i\right)+\BGs^0_1:\BC_*^{-1}\BGs^0_1 \right],
\eeq{0.17}
approaches zero as the constituent stiff isotropic material becomes increasingly stiff and 
the constituent compliant isotropic material becomes increasingly compliant.
The lattice structure of a pentamode is similar to that of diamond with a stiff double cone structure replacing each carbon bond. 
This structure ensures that the tips of four double cone structures
meet at each vertex. This is the essential feature: treating the double cone structures as strut, the tension in one determines uniquely the tension in the
other three. This is simply balance of forces. Thus the structure as a whole can essentially only support one stress.
Pentamode structures were experimentally realized by Kadic, B{\"u}ckmann, Stenger, Thiel and Wegener \cite{Kadic:2012:PPM}. Pentamode structures were also independently discovered in 1995 by Sigmund, although he
did not find the complete span of pentamode structures needed here: one needs pentamodes that can support
any chosen stress, not just a hydrostatic one. It is this aspect of pentamodes that makes them more interesting
than, for example, a gel. 
in an incredible feat of precision three dimensional lithography. One of their electron micrographs of the structure is shown in \fig{0}.
Gels are examples of pentamodes as they are easy to shear, but difficult to compress under a hydrostatic loading $\BGs_1=\BI$. By contrast the pentamodes of 
Milton and Cherkaev could be stiff to any desired stress $\BGs^0_1$: this desired stress may be a mixture of shear and compression, and may have eigenvalues of mixed signs.
A simple argument for seeing that these pentamodes can achieve any desired elasticity tensor was given in the foreword of the book edited by Phani and Hussein \cite{Phani:2016:DLM}. 
To recapitulate that argument, one expresses the desired $\BC_*$ in terms of its eigenvectors and eigenvalues,
\beq \BC_*=\sum_{i=1}^6\Gl_i\Bv_i\otimes\Bv_i. \eeq{2.10}
The idea, roughly speaking, is to find 6 pentamode structures each supporting a stress represented by the vector $\Bv_i$, $i=1,2,\ldots, 6$. The stiffness
of the material and the necks of the junction regions at the vertices need to be adjusted so each pentamode structure has an effective elasticity tensor
close to
\beq \BC_*^{(i)}=\Gl_i\Bv_i\otimes\Bv_i. \eeq{2.11}
Then one successively superimposes all these 6 pentamode structures, with their lattice structures being offset to avoid collisions. Additionally
one may need to deform the structures appropriately to avoid these collisions as described in \cite{Milton:1995:WET}, and when one does this it is necessary
to readjust the stiffness of the material in the structure to maintain the value of $\Gl_i$.
Then the remaining void in the structure is replaced by an extremely compliant
material. (Its presence is just needed for technical reasons, to ensure that the assumptions of homogenization theory are valid so that the elastic
properties can be described by an effective tensor.)  But it is so compliant that essentially  the effective elasticity tensor is just a sum of the
effective elasticity tensors of the 6 pentamodes, i.e., the elastic interaction between the 6 pentamodes is negligible. In this way we arrive at a material with (approximately) 
the desired elasticity tensor $\BC_*$. 

\begin{figure}[!ht]
\centering
\includegraphics[width=0.7\textwidth]{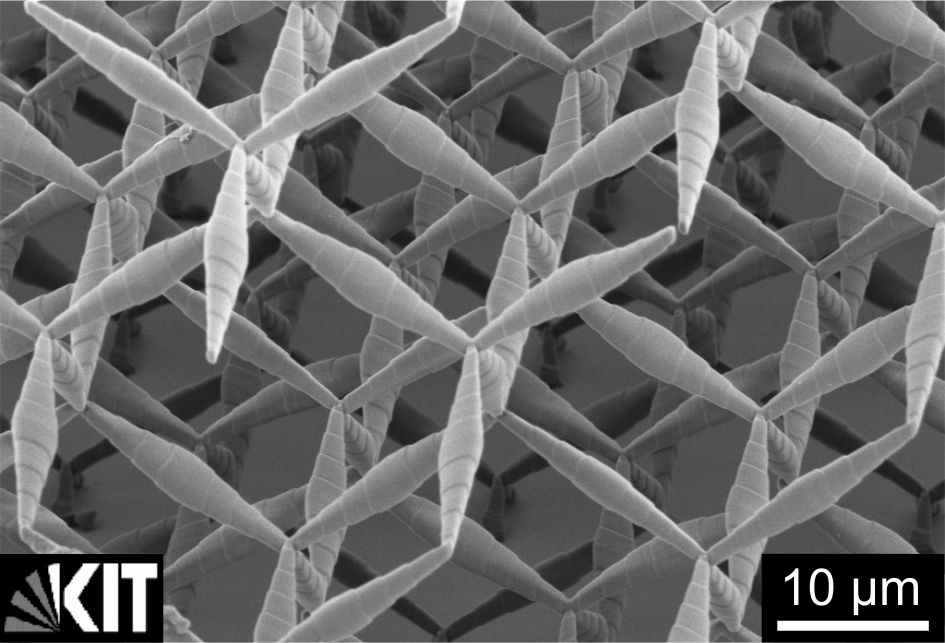}
\caption{An electron micrograph of the pentamode structure created by Kadic, B{\"u}ckmann, Stenger, Thiel and Wegener \protect\cite{Kadic:2012:PPM} using a three dimensional lithography technique. Used with the kind permission of Martin Wegener.}
\labfig{0}
\end{figure}

It is worth mentioning that with extremely high contrast materials the homogenized equations are not necessarily the usual linear elasticity equations, but can also include
nonlocal terms. Nonlocal interactions can be obtained for example with an extremely stiff dumbbell shaped inclusion with the balls arbitrarily distant. If the bar joining them
is not only extremely stiff but also extremely thin, then it does directly not couple with the surrounding elastic material (except in the very near vicinity of the bar, where it is
obviously deformed by it), but provides a nonlocal interaction between the balls. In fact, amazingly, Camar-Eddine and Seppecher \cite{Camar:2003:DCS} have completely characterized all possible
linear macroscopic behaviors of any high contrast composite: they show that any energetically stable behavior can be obtained
using materials with such dumbbell shaped inclusions interacting at many length scales. Some interesting examples of high contrast materials
with exotic effective behaviors have been given by Seppecher,  Alibert, and Isola \cite{Seppecher:2011:LET}.
\section{Characterizing Convex Sets and $G$-closures for elasticity}
\setcounter{equation}{0}
Let $G$ be a convex set of real $d$-dimensional vectors, meaning that if $\Bc_1,\Bc_2\in G$ then $\Gt\Bc_1+(1-\Gt)\Bc_2\in G$ for all $\Gt\in[0,1]$.  
As shown in \fig{1}(a)
for $d=2$ such a convex set can be completely characterized by its Legendre transform,
\beq f(\Bn)=\min_{\Bc\in G}\Bn\cdot\Bc. \eeq{1.1}
Clearly this function satisfies the homogeneity property that
\beq f(\Gl\Bn)=\Gl f(\Bn)\quad {\rm for~all~}\Gl>0, \eeq{1.2}
and consequently it suffices to know $f(\Bn)$ for all unit vectors $\Bn$ to recover the function $f(\Bn)$ for any vector $\Bn$. The values of $f(\Bn)$ and $f(-\Bn)$
give the position of the two planes with normals $\pm\Bn$ that are tangent to $G$: specifically $|f(\Bn)|$ and $|f(-\Bn)|$ give the distances from these tangent
planes to the origin.
By varying $\Bn$ and taking the intersection of the regions between the planes one
recovers $G$: the set $G$ is the envelope of its tangent planes as illustrated in \fig{1}(a). 
Thus the Legendre transform function $f(\Bn)$, with $|\Bn|=1$ 
completely characterizes $G$. 

\begin{figure}[!ht]
\centering
\includegraphics[width=0.7\textwidth]{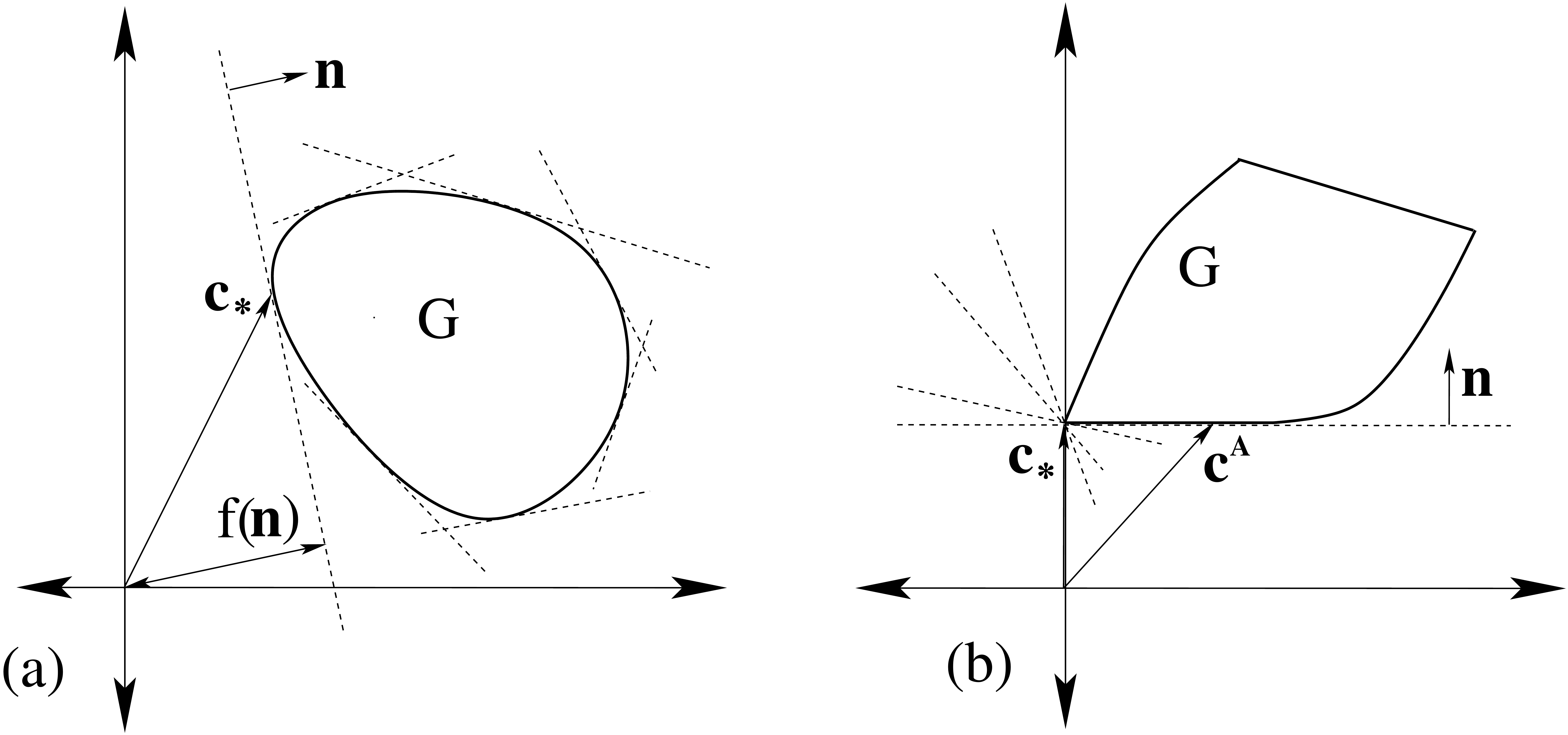}
\caption{As shown in (a) a convex set is the envelope of its tangent planes. The position of the two tangent planes with normal $\Bn$ are determined by the
Legendre transform $f(\Bn)$ and $f(-\Bn)$ defined by \eq{1.1}. Specifically $f(\Bn)$ and $f(-\Bn)$ give the distances of the tangent planes from the origin.
The example of (b) highlights an interesting case discussed in the text that helps give a geometrical interpretation of the results of the paper.}
\labfig{1}
\end{figure}

The example of \fig{1}(b) is also illuminating for the purposes of this paper. Let $\Bn$ and $\Bm$ be the vectors
\beq \Bn=\bpm 0 \\ 1 \epm,\quad\quad \Bm=\bpm 1 \\ 0\epm, \eeq{1.3}
and consider $f(\Bn+\Ga\Bm)$ for $\Ga\geq 0$ in the context of this example. (Of course $\Bn+\Ga\Bm$ is only a unit vector when $\Ga=0$). As the boundary of $G$
contains a flat section orthogonal to $\Bn$, the vector $\Bc$ which attains the minimum in \eq{1.1} is nonunique. In the diagram both $\Bc^A$ and $\Bc_*$ are minimizers.
However for an infinitesimal value of $\Ga>0$, $\Bc_*$ is selected as the unique minimizer and remains the minimizer no matter how large $\Ga>0$ becomes. Furthermore,
since $\Bc_*$ is orthogonal to $\Bm$ the value of $f(\Bn+\Ga\Bm)$ remains constant for all $\Ga\geq 0$.

If $G$ is a convex set of say real $d\times d$ matrices it can be similarly characterized by its Legendre transform,
\beq f(\BN)=\min_{\BC\in G}(\BN,\BC), \eeq{1.4}
defined for all $d\times d$ matrices $\BN$, where $(\BN,\BC)$ is an inner product on the space of matrices which we may take to be
\beq (\BN,\BC)=N_{ij}C_{ij}\equiv \BN:\BC, \eeq{1.5}
where we have adopted the Einstein summation convention that sums over repeated indices are assumed, and the double dot ``:'' denotes a double contraction of indices. 
This is exactly equivalent to \eq{1.1} if we think of the
matrix $\BC$ being represented by the vector $\Bc$ of its matrix elements. Note that if $G$ only contains symmetric matrices, then it suffices to take $\BN$ as
a symmetric matrix since $(\BA,\BC)=0$ if $C$ is symmetric and $A$ is antisymmetric.

Similarly if $G$ is a convex set of fourth order elasticity tensors $\BC$ satisfying the usual symmetries
\beq C_{ijk\ell}=C_{jik\ell}=C_{k\ell ij}, \eeq{1.6}
then it can be characterized by the Legendre transform \eq{1.4} with an inner product
\beq (\BN,\BC)=N_{ijk\ell}C_{ijk\ell}, \eeq{1.7}
and again it suffices to assume $\BN$ has the same symmetries as $\BC$, i.e., those given by \eq{1.6}.

However, $G$-closures (i.e., sets of all possible effective tensors) are not generally convex sets. Nevertheless, they do have some convexity properties
as a consequence of their stability under lamination. In the case
of the set $GU_f$ where $U=\{\BC_1,\Gd\BC_2\}$, we can take two materials with effective tensors $\BC_1^*,\BC_2^*\in GU_f$ and laminate them together in a direction $\Bn$
(representing the vector perpendicular to the layers) in proportions $\Gt$ and $1-\Gt$ to obtain an effective tensor $\BC_*(\Bn,\Gt)$ which necessarily lies in the
set $GU_f$ for all $\Gt\in[0,1]$. While $\BC_*(\Bn,\Gt)$ is not a linear average of  $\BC_1^*$ and $\BC_2^*$, there exist fractional linear transformations $T_{\Bn}$ 
of fourth order tensors such that lamination in direction $\Bn$ reduces to a linear average \cite{Backus:1962:LWE, Milton:1990:CSP} (see also \cite{Tartar:1979:ECH}),
\beq T_{\Bn}(\BC_*(\Bn,\Gt))=\Gt T_{\Bn}(\BC_1^*)+(1-\Gt)T_{\Bn}(\BC_2^*)\quad{\rm for~all~}\Gt\in[0,1]. \eeq{1.8}
Thus $T_{\Bn}(GU_f)$ must be a convex set of fourth order tensors. In the particular case where a set of effective tensors has no interior, i.e., is constrained 
to lie on a manifold of dimension $m$ smaller than the dimension of the space of  fourth order tensors satisfying the symmetries of elasticity tensors (i.e., $m<21$ , 
for three-dimensional composites and $m<6$ for two-dimensional composites)
then as recognized by Grabovsky \cite{Grabovsky:1998:EREa} (see also \cite{Grabovsky:1998:EREb}) $T_{\Bn}$ must map this manifold to a subset
of a hyperplane of dimension $m$ for any value of $\Bn$. 
This places rather severe constraints on the form of such
manifolds. Identifying such manifolds is important as they represent exact relations satisfied by effective tensors, no matter what the geometry of the composite happens to be.
Thus these constraints provide necessary conditions for an exact relation. Later, sufficient conditions for an exact relation to hold were obtained \cite{Grabovsky:2000:ERE}.

Unfortunately, the use of Legendre transforms of the convex set $T_{\Bn}(GU_f)$ is not useful to us as we are unaware of any direct variational principles for $T_{\Bn}(\BC_*)$.
An alternative approach was prompted by work of Cherkaev and Gibiansky \cite{Cherkaev:1992:ECB, Cherkaev:1993:CEB} who found that bounding sums of energies and complementary energies could lead to very useful
bounds on $G$-closures. It was proved by Francfort and Milton \cite{Francfort:1994:SCE, Milton:1994:LBS} 
that minimums over $\BC_*\in GU_f$ of such sums of energies and complementary energies completely characterize
$GU_f$ in much the same way that Legendre transforms characterize convex sets: the stability under lamination of $GU_f$ is what allows one to recover $GU_f$ from the values
of these minimums (see also Chapter 30 in \cite{Milton:2002:TOC}). \fig{30.3} captures the idea of this characterization.

\begin{figure}[!ht]
\centering
\includegraphics[width=0.7\textwidth]{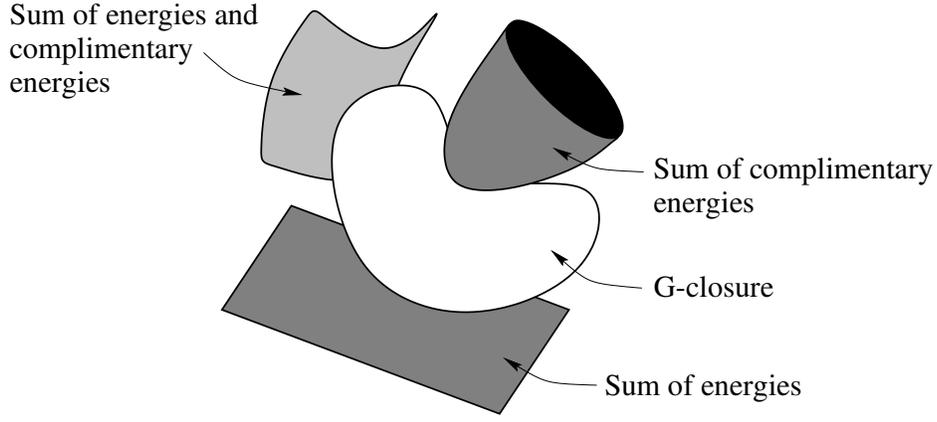}
 \caption{$G$-closures are characterized by minimums of sums
of energies and complementary energies. The coordinates here represent
the elements of the effective elasticity tensor $\protect\BC_*$. Then a plane represents a
surface where a sum of energies is constant, and when this sum takes its
minimum value the plane is tangent to the $G$-closure. The convexity
properties of the $G$-closure guarantee that the surfaces corresponding
to the minimums of sums of energies and complementary energies wrap around the $G$-closure
and touch each point on its boundary. Reproduction of figure 30.1 in \protect\cite{Milton:2002:TOC}}
\labfig{30.3}
\end{figure}

Specifically, in the case of three-dimensional elasticity, the set $GU_f$ is completely characterized if we know the $7$ ``energy functions'',
\beqa W_f^0(\BGs^0_1,\BGs^0_2,\BGs^0_3,\BGs^0_4,\BGs^0_5,\BGs^0_6)
& = &\min_{\BC_*\in GU_f}\sum_{j=1}^6\BGs^0_j:\BC_*^{-1}\BGs^0_j, \nonum
 W_f^1(\BGs^0_1,\BGs^0_2,\BGs^0_3,\BGs^0_4,\BGs^0_5,\BGe^0_1)
& = &\min_{\BC_*\in GU_f}\left[\BGe^0_1:\BC_*\BGe^0_1+\sum_{j=1}^5\BGs^0_j:\BC_*^{-1}\BGs^0_j\right],\nonum
 W_f^2(\BGs^0_1,\BGs^0_2,\BGs^0_3,\BGs^0_4,\BGe^0_1,\BGe^0_2)
& = &\min_{\BC_*\in GU_f}\left[\sum_{i=1}^2\BGe^0_i:\BC_*\BGe^0_i+\sum_{j=1}^4\BGs^0_j:\BC_*^{-1}\BGs^0_j\right],\nonum
 W_f^3(\BGs^0_1,\BGs^0_2,\BGs^0_3,\BGe^0_1,\BGe^0_2,\BGe^0_3)
& = &\min_{\BC_*\in GU_f}\left[\sum_{i=1}^3\BGe^0_i:\BC_*\BGe^0_i+\sum_{j=1}^3\BGs^0_j:\BC_*^{-1}\BGs^0_j\right],\nonum
 W_f^4(\BGs^0_1,\BGs^0_2,\BGe^0_1,\BGe^0_2,\BGe^0_3,\BGe^0_4)
& = &\min_{\BC_*\in GU_f}\left[\sum_{i=1}^4\BGe^0_i:\BC_*\BGe^0_i+\sum_{j=1}^2\BGs^0_j:\BC_*^{-1}\BGs^0_j\right],\nonum
 W_f^5(\BGs^0_1,\BGe^0_1,\BGe^0_2,\BGe^0_3,\BGe^0_4,\BGe^0_5)
& = & \min_{\BC_*\in GU_f}\left[\left(\sum_{i=1}^5\BGe^0_i:\BC_*\BGe^0_i\right)+\BGs^0_1:\BC_*^{-1}\BGs^0_1 \right],\nonum
 W_f^6(\BGe^0_1,\BGe^0_2,\BGe^0_3,\BGe^0_4,\BGe^0_5,\BGe^0_6)
& = & \min_{\BC_*\in GU_f}\sum_{i=1}^6\BGe^0_i:\BC_*\BGe^0_i.
\eeqa{2.1}
In fact, it suffices \cite{Milton:1995:WET}
to know these functions for sets of applied strains $\BGe^0_i$ and applied stresses $\BGs^0_j$ that are mutually orthogonal:
\beqa (\BGe^0_i,\BGs^0_j)& = & 0,\quad (\BGe^0_i,\BGe^0_k)=0,\quad (\BGs^0_j,\BGs^0_\ell)=0 \quad{\rm for~all~}i, j, k,\ell~{\rm with}~i\ne j, ~
i\ne k, ~j\ne \ell.\nonum &~&
\eeqa{2.1aa}
Each of these terms in the minimums has a physical significance. For example, in the expression for $W_f^2$,
\beq \sum_{i=1}^2\BGe^0_i:\BC_*\BGe^0_i+\sum_{j=1}^4\BGs^0_j:\BC_*^{-1}\BGs^0_j \eeq{2.1a}
has the physical interpretation as being the sum of energies per unit volume stored in the composite with effective elasticity tensor $\BC_*$
when it is subjected to successively the two applied strains $\BGe^0_1$ and $\BGe^0_2$ and then to the four applied stresses $\BGs^0_1$, $\BGs^0_2$, $\BGs^0_3$,
$\BGs^0_4$. To distinguish the terms $\BGe^0_i:\BC_*\BGe^0_i$ and $\BGs^0_j:\BC_*^{-1}\BGs^0_j$, the first is called an energy (it is really an energy per unit volume
associated with the applied strain $\BGe^0_i$) and the second is called a complementary energy, although it too physically represents an energy per unit volume
associated with the applied stress $\BGs^0_j$. Note that the quantity \eq{2.1a} can be equivalently written as 
\beq (\BC_*,\BN)+(\BC_*^{-1},\BN'), \eeq{2.1b}
where
\beq \BN=\sum_{i=1}^2\BGe^0_i\otimes\BGe^0_i,\quad \BN'=\sum_{j=1}^4\BGs^0_j\otimes\BGs^0_j,
\eeq{2.1c}
in which for any $d\times d$ symmetric matrix $\BA$, the tensor $\BA\otimes\BA$ is defined to be fourth order tensor with elements
\beq  \{\BA\otimes\BA\}_{ijk\ell}=\{\BA\}_{ij}\{\BA\}_{k\ell}. \eeq{2.1d}
If we decompose the positive semidefinite tensors $\BN$ and $\BN'$ into their spectral decompositions
\beq \BN=\sum_{i=1}^2\Gl_i\Bv_i\otimes\Bv_i,\quad\quad \BN'=\sum_{j=1}^4\Gl_j'\Bv_j'\otimes\Bv_j', \eeq{2.1e}
with eigenmatrices $\Bv_i$ and $\Bv_j'$ and corresponding nonnegative eigenvalues $\Gl_i$ and $\Gl_j'$
then, with the orthogonality constraints \eq{2.1aa}, we can make the identifications
\beq \BGe^0_i=\sqrt{\Gl_i}\Bv_i,\quad \BGs^0_j=\sqrt{\Gl_j'}\Bv_j. \eeq{2.1f}
Note that due to the orthogonality conditions \eq{2.1aa} the fourth-order tensors $\BN$ and $\BN'$ have the property that the product $\BN\BN'$ is zero. Here the
product of two fourth order tensors $\BC$ and $\BC'$ is given by
\beq \{\BC\BC'\}_{ijk\ell}=\{\BC\}_{ijmn}\{\BC'\}_{mn  k\ell}.
\eeq{2.1g}
Thus in the same way that convex sets are the envelope of planes, the $G$-closure $GU_f$ is the envelope of special surfaces parameterized by positive
semidefinite fourth order tensors $\BN$ and $\BN'$ satisfying the symmetries of elasticity tensors, and having zero product $\BN\BN'=\BN'\BN=0$ (i.e.,
the range of $\BN'$ is in the null space of $\BN$, and conversely the range of $\BN$ is in the null space of $\BN'$). These special surfaces consist of all 
positive definite fourth order
tensors $\BC$ satisfying
\beq (\BC,\BN)+(\BC^{-1},\BN')=c, \eeq{2.1h}
where $c$ is a positive real constant. In the case $\BN'=0$ this does represent a hyperplane, but its orientation is restricted by
the fact that the outward normal to the surface $\BN$ is restricted to be a positive definite fourth order tensor (by outward normal we mean the normal pointing 
away from the origin). Knowledge of the seven functions $W_f^i$ given by \eq{2.1} is clearly equivalent to knowledge
of the function
\beq W_f(\BN,\BN')=\min_{\BC_*\in GU_f}(\BC_*,\BN) + (\BC_*^{-1},\BN'),
\eeq{2.1i}
for all positive semidefinite fourth order tensors $\BN$ and $\BN'$ satisfying the symmetries of elasticity tensors, and having zero product $\BN\BN'=0$.
The formula for recovering $GU_f$ from $W_f(\BN,\BN')$ is then
\beq \bigcap_{\substack{\BN,\BN'\geq 0 \\ \BN\BN'=0}}\{\BC: (\BC,\BN)+(\BC^{-1},\BN')\geq W_f(\BN,\BN')\}=GU_f.
\eeq{2.1j}
More generally if we replace $GU_f$ in \eq{2.1i} by another set $G$ of positive definite matrices, and if the left hand side of \eq{2.1j} is again $G$, then we may say $G$ is ``W-convex''.

An explicit definition of ``W-convexity'', is as follows: a set $G$ of positive definite
symmetric matrices is said to be ``strictly W-convex" if $G$ is simply connected and if for every pair of positive semidefinite symmetric matrices $\BN$ and $\BN'$, not both zero, the minimum in
\beq \min_{\BC\in G}(\BC,\BN)+(\BC^{-1},\BN') \eeq{2.1ja}
is uniquely attained by only one matrix $\BC\in G$. Geometrically, $G$ is ``strictly W-convex" if for all positive semidefinite symmetric matrices $\BN$ and $\BN'$, not both zero, the surface
that consists of all positive definite matrices $\BC$ satisfying  
\beq (\BC,\BN)+(\BC^{-1},\BN')=k \eeq{2.1jb}
where $k$ is chosen as the smallest value for which this surface touches $G$, has the property that it touches $G$ at only one point. A set $G$ is ``W-convex" if it is a limit of ``strictly W-convex" sets.
If the set $G$ has a smooth boundary, then the condition for `W-convexity'' can be expressed in terms of the
curvature of the boundary of $G$: when $G$ is a set of matrices, this curvature at each point on the surface
of $G$ is a fourth-order tensor;
when $G$ is a set of fourth order elasticity tensors, this curvature is an eighth order tensor. 
(See equation (3.51) in \cite{Milton:1994:LBS}, or equation (30.11) in \cite{Milton:2002:TOC}, for the explicit inequalities that the curvature must satisfy).

The stability of $GU_f$ under lamination implies it is $W$-convex, but $W$-convexity probably does not imply stability under lamination, as stability under lamination
depends on the underlying partial differential equations. Associated with any set $G$ of symmetric positive-definite matrices $\BC$ is its $W$-transform, defined as
\beq  W(\BN,\BN')=\min_{\BC\in G}(\BC,\BN)+(\BC^{-1},\BN'),
\eeq{2.1k}
where $\BN$ and $\BN'$ are symmetric positive semidefinite matrices, satisfying $\BN\BN'=0$, and the inner product of two symmetric matrices $\BA$ and $\BB$ can be taken as
$(\BA,\BB)=\Tr(\BA\BB)$ where $\Tr$ denotes the trace (sum of diagonal elements) of a matrix. To see some of the properties of $W$-transforms it is helpful to extend the definition of the transform to allow for matrices $\BN$ and $BN'$ that have a non-zero product, $\BN\BN'\ne 0$. The defining equation, \eq{2.1k}, remains the same. 
Then consider a weighted average of $(\BN_1,\BN_1')$ and  $(\BN_2,\BN_2')$, with weights $\Gt$ and $1-\Gt$, where the four matrices
$\BN_1, \BN_1', \BN_2,\BN_2'$ are positive semidefinite. Then for and $\Gt\in(0,1)$, we have
\beqa &~& W(\Gt\BN_1+(1-\Gt)\BN_2,\Gt\BN'_1+(1-\Gt)\BN'_2) \nonum
& ~ & \quad = \min_{\BC\in G}\left\{\Gt[(\BC,\BN_1)+(\BC^{-1},\BN'_1)]+(1-\Gt)[(\BC,\BN_2)+(\BC^{-1},\BN'_2)]\right\},\nonum
& ~ & \quad \geq \Gt\left\{\min_{\BC\in G}(\BC,\BN_1)+(\BC^{-1},\BN'_1)\right\}+(1-\Gt)\left\{\min_{\BC\in G}(\BC,\BN_2)+(\BC^{-1},\BN'_2)\right\} \nonum
& ~ & \quad \geq \Gt W_f(\BN_1,\BN'_1)+(1-\Gt) W_f(\BN_2,\BN'_2),
\eeqa{2.1m}
which (by definition) implies $W(\BN,\BN')$ is a jointly concave function of $\BN$ and $\BN'$. This concavity is a well known property
of Legendre transforms.

\section{Variational Principles}
\setcounter{equation}{0}

Upper bounds on the sums of energies and complementary energies can easily be obtained from classic energy minimization variational principles. For example, in
the case of the sum \eq{2.1a}, we have
\beqa &~&\sum_{i=1}^2\BGe^0_i:\BC_*\BGe^0_i+\sum_{j=1}^4\BGs^0_j:\BC_*^{-1}\BGs^0_j=\nonum
&~&~\min_{\und{\BGe}_1,\und{\BGe}_2,\und{\BGs}_1,\und{\BGs}_2,\und{\BGs}_3,\und{\BGs}_4}
\Big\langle\sum_{i=1}^2\und{\BGe}_i(\Bx):\BC(\Bx)\und{\BGe}_i(\Bx)+\sum_{j=1}^4\und{\BGs}_j(\Bx):[\BC(\Bx)]^{-1}\und{\BGs}_j(\Bx)\Big\rangle \nonum
&~&
\eeqa{2.2}
where the minimum is over a set of two trial strain fields $\und{\BGe}_1(\Bx)$ and $\und{\BGe}_2(\Bx)$,
and a set of four trial strain fields $\und{\BGs}_1(\Bx)$, $\und{\BGs}_2(\Bx)$, $\und{\BGs}_3(\Bx)$, and $\und{\BGs}_4(\Bx)$
that have the prescribed average values,
\beq \lang\und{\BGe}_i\rang=\BGe^0_i~~{\rm for}~i=1,2,\quad\lang\und{\BGs}_j\rang=\BGs^0_j~~{\rm for}~j=1,2,3,4,
\eeq{2.3}
and are subject to the differential constraints that
\beqa &~&\und{\BGe}_i(\Bx)=[\Grad\und{\Bu}_i(\Bx)+(\Grad\und{\Bu}_i(\Bx))^T]/2~~{\rm for}~i=1,2,\nonum
&~&\Div\und{\BGs}_j(\Bx)=0~~{\rm for}~j=1,2,3,4, \eeqa{2.4}
where $T$ denotes the transpose (reflecting the matrix about its diagonal) and
$\und{\Bu}_i(\Bx)$ is the trial displacement field associated with the trial stress field $\und{\BGe}_i(\Bx)$. The trial strain fields $\und{\BGe}_i(\Bx)$ and the trial
stress fields $\und{\BGs}_j(\Bx)$ (but not the trial displacement fields) should be chosen to be periodic (if the composite is periodic), quasiperiodic 
(if the composite is quasiperiodic), or statistically homogeneous (if the composite is statistically homogeneous). 
It may be the case that the material has structure on widely separated length scales. Maybe it can be viewed as a mixture of two composites, one with effective tensor $\BC_*^{1}$
and a second with effective tensor $\BC_*^{2}$ so that at the mesoscale it has a geometry described by a characteristic function $\chi_*(\Bx)$, where $\chi_*(\Bx)$ is one in the
composite with effective tensor $\BC_*^{1}$ and zero in the material with effective tensor $\BC_*^{2}$. Naturally the length scale, or length scales, or variations in $\chi_*(\Bx)$
should be much larger than the variations in the microstructure of the materials that have the effective tensors $\BC_*^{1}$ and $\BC_*^{2}$. Then we can treat the material
having effective tensor as a composite of the materials $\BC_*^{1}$ and $\BC_*^{2}$ and we have the variational principle
\beqa &~&\sum_{i=1}^2\BGe^0_i:\BC_*\BGe^0_i+\sum_{j=1}^4\BGs^0_j:\BC_*^{-1}\BGs^0_j=\nonum
&~&\min_{\und{\BGe}_1,\und{\BGe}_2,\und{\BGs}_1,\und{\BGs}_2,\und{\BGs}_3,\und{\BGs}_4}
\Big\langle\sum_{i=1}^2\und{\BGe}_i(\Bx):[\chi_*(\Bx)\BC_*^{1}+(1-\chi_*(\Bx))\BC_*^{2}]\und{\BGe}_i(\Bx)\nonum
&~&+\sum_{j=1}^4\und{\BGs}_j(\Bx):[\chi_*(\Bx)\BC_*^{1}+(1-\chi_*(\Bx))\BC_*^{2}]^{-1}\und{\BGs}_j(\Bx)\Big\rangle,
\eeqa{2.5}
where again the minimum is over fields subject to the appropriate average values and differential constraints. Particular choices of trial fields 
will then lead to an upper bound on this
sum of energies and complementary energies. To bound the quantities on the right one may again use variational principles. When $\Bx$ is in the material $\BC_*^{k}$, $k=1$ or $2$, 
one has the variational principles
\beqa &~& \und{\BGe}_i(\Bx):\BC_*^{k}\und{\BGe}_i(\Bx)=\min_{\dund{\BGe}_i}\lang\dund{\BGe}_i(\Bx,\By):\BC^{k}(\By)\dund{\BGe}_i(\Bx,\By)\rang_{\By}, \nonum
&~& \und{\BGs}_j(\Bx):[\BC_*^{k}]^{-1}\und{\BGs}_j(\Bx)=\min_{\dund{\BGs}_j}\lang\dund{\BGs}_j(\Bx,\By):[\BC^{k}(\By)]^{-1}\dund{\BGs}_j(\Bx,\By)\rang_{\By},
\eeqa{2.6}
where $\lang\cdot\rang_{\By}$ now denotes an average over the $\By$ variable ($\Bx$ is the ``slow variable'' and $\By$ is the ``fast variable'') and 
\beq \BC^{k}(\By)=\chi^k(\By)\BC_1+(1-\chi^k(\By))\BC_2,
\eeq{2.7}
in which $\chi^k(\By)$ is the characteristic function representing the geometry associated with the effective tensor $\BC_*^{k}$, taking a value 1 in the material with tensor $\BC_1$ and
0 in the material with tensor $\BC_2$.  Here the trial fields have the prescribed average values,
\beq \lang\dund{\BGe}_i(\Bx,\By)\rang_{\By}=\und{\BGe}_i(\Bx)~~{\rm for}~i=1,2,\quad\lang\dund{\BGs}_j(\Bx,\By)\rang_{\By}=\und{\BGs}_j(\Bx)~~{\rm for}~j=1,2,3,4,
\eeq{2.8}
and are subject to the differential constraints that
\beqa &~&\dund{\BGe}_i(\Bx,\By)=[\Grad_y\dund{\Bu}_i(\Bx,\By)+(\Grad\dund{\Bu}_i(\Bx,\By))^T]/2~~{\rm for}~i=1,2,\nonum
&~&\Div_y\dund{\BGs}_j(\Bx,\By)=0~~{\rm for}~j=1,2,3,4, \eeqa{2.9}
where $\Grad_y$ and $\Div_y$ are the gradient and divergence with respect to the $\By$ variables. We call the step of replacing the variational principle
\eq{2.2} by the variational principles \eq{2.5} and \eq{2.6} the ``homogenization at intermediate scales step''.

In this paper we will choose trial fields that satisfy the {\it local orthogonality condition} that  
\beq \und{\BGe}_i(\Bx):\und{\BGs}_j(\Bx)=0,\quad {\rm for~all~}\Bx. \eeq{2.9a}
Using the differential constraints satisfied by the trial fields, and integration by parts, one sees that the associated average fields are necessarily orthogonal too:
\beq  \BGe^0_i:\BGs^0_j=\lang\und{\BGe}_i(\Bx)\rang:\lang\und{\BGs}_j(\Bx)\rang=\lang\und{\BGe}_i(\Bx):\und{\BGs}_j(\Bx)\rang=0. \eeq{2.9b}

\section{Finding most of the energy functions}
\setcounter{equation}{0}

Recall from Section \sect{review} that an complementary Avellaneda material is a sequentially layered laminate material with phase 1 occupying a volume fraction $f$ and with effective tensor
\[
\widetilde{\BC}_f^A(\BGs^0_1,\BGs^0_2,\BGs^0_3,\BGs^0_4,\BGs^0_5,0)
\]
that attains equality in \eq{0.12}. It is found by minimizing the right hand side of \eq{0.12} as $\BC_*$ varies within the class of tensors given by \eq{0.13.0}-\eq{0.13.2} with $\BC_2=0$, as the 
rank $r$, the positive weights $c_j$ which sum to $1$, and the unit vectors $\Bn_i$ are varied. Here some of the applied stresses $\BGs^0_j$ could be zero.
Since the energy $\BGs^0_j:\BC_*^{-1}\BGs^0_j$ associated with any applied stress $\BGs^0_j$  is necessarily nonnegative, we obtain from \eq{2.1} the bounds  
\beqa
\sum_{j=1}^5\BGs^0_j:[\widetilde{\BC}_f^A(\BGs^0_1,\BGs^0_2,\BGs^0_3,\BGs^0_4,\BGs^0_5,0)]^{-1}\BGs^0_j & \leq & W_f^1(\BGs^0_1,\BGs^0_2,\BGs^0_3,\BGs^0_4,\BGs^0_5,\BGe^0_1), \nonum
\sum_{j=1}^4\BGs^0_j:[\widetilde{\BC}_f^A(\BGs^0_1,\BGs^0_2,\BGs^0_3,\BGs^0_4,0,0)]^{-1}\BGs^0_j & \leq & W_f^2(\BGs^0_1,\BGs^0_2,\BGs^0_3,\BGs^0_4,\BGe^0_1,\BGe^0_2), \nonum
\sum_{j=1}^3\BGs^0_j:[\widetilde{\BC}_f^A(\BGs^0_1,\BGs^0_2,\BGs^0_3,0,0,0)]^{-1}\BGs^0_j & \leq & W_f^3(\BGs^0_1,\BGs^0_2,\BGs^0_3,\BGe^0_1,\BGe^0_2,\BGe^0_3), \nonum
\sum_{j=1}^2\BGs^0_j:[\widetilde{\BC}_f^A(\BGs^0_1,\BGs^0_2,0,0,0,0)]^{-1}\BGs^0_j &\leq & W_f^4(\BGs^0_1,\BGs^0_2,\BGe^0_1,\BGe^0_2,\BGe^0_3,\BGe^0_4), \nonum
\BGs^0_1:[\widetilde{\BC}_f^A(\BGs^0_1,0,0,0,0,0)]^{-1}\BGs^0_1 & \leq & W_f^5(\BGs^0_1,\BGe^0_1,\BGe^0_2,\BGe^0_3,\BGe^0_4,\BGe^0_5), \nonum
0&\leq & W_f^6(\BGe^0_1,\BGe^0_2,\BGe^0_3,\BGe^0_4,\BGe^0_5,\BGe^0_6).
\eeqa{2.12}
The last inequality is clearly sharp, being attained when the composite consists of islands of phase 1 surrounded by a phase 2 (so that $\BC_*$ approaches $0$ as $\Gd\to 0$).
The objective of this paper is to show that many of the other inequalities are sharp too in the limit $\Gd\to 0$ at least when the spaces spanned by the applied strains
$\BGe^0_j$, $j=1,2,\ldots,p$ satisfy certain properties. This space of applied strains $\CV_p$, associated with $W_f^p$, has dimension $p$ and is spanned by $\BGe^0_1,\BGe^0_2,\ldots,\BGe^0_p$.

The recipe for doing this is to simply insert into a relevant complementary Avellaneda material a microstructure occupying a thin walled region, such that the material can slip along the walls when the 
applied strain lies in appropriate spaces $\CV_p$, yet which is such that the combination of Avellaneda material and walled material can support without slip any applied stress in the subspace orthogonal to $\CV_p$.
This will be possible only when $\CV_p$ is spanned by symmetrized rank one matrices, taking the form
\beq \BGe^{(k)}=(\Ba_{k}\Bn^T_k+\Bn_k\Ba_{k}^T)/2,\quad {\rm for}~k=1,\ldots,p. \eeq{2.12a}
The existence of such matrices $\BGe^{(k)}$ is proved in Section~\sect{symr1mat}. The proof uses small perturbations of the applied stresses and strains.
But, due the continuity of the energy functions $W_f^k$ established in Section~\sect{continuity}, the small perturbations do not modify the generic result.
The vectors $\Bn_k$ determine the orientation of the walls in the structure. For each $\Bn_k$ there is a set of parallel walls perpendicular to $\Bn_k$ that allow
slip given by the strain $\BGe^{(k)}$. We say slip but it should be recognized that $\BGe^{(k)}$ is not generally a pure shear, but rather a combination
of dilation and shear, since it does not generally have zero trace. 

To define the thin walled structure, introduce the periodic function $H_c(x)$ with period $1$ which takes the value $1$ 
if $x-[x]\leq c$ where $[x]$ is the greatest integer less than $x$, and $c\in[0,1]$ gives the thickness of each wall relative
to the spacing between walls (which is unity). Then for the unit vectors
$\Bn_1,\Bn_2,\ldots,\Bn_p$ appearing in \eq{2.12a}, and for a small relative wall thickness $c=\Ge$ define the characteristic functions 
\beq \Gn_k(\Bx)=H_\Ge(\Bx\cdot\Bn_k+k/p). \eeq{2.13}
This characteristic function defines a series of parallel walls, as shown in \fig{2}(a),
each perpendicular to the vector $\Bn_j$, where $\Gn_j(\Bx)=1$ in the wall material.  The additional shift term $k/p$ in \eq{2.13} ensures the walls
associated with $k_1$ and $k_2$ do not intersect when it happens that $\Bn_{k_1}=\Bn_{k_2}$, at least when $\Ge$ is small. Note that $\Ge$ is a volume fraction, not a homogenization parameter. We will be taking the limit $\Ge\to 0$
{\it after} taking the homogenization limit.


\begin{figure}[!ht]
\centering
\includegraphics[width=0.7\textwidth]{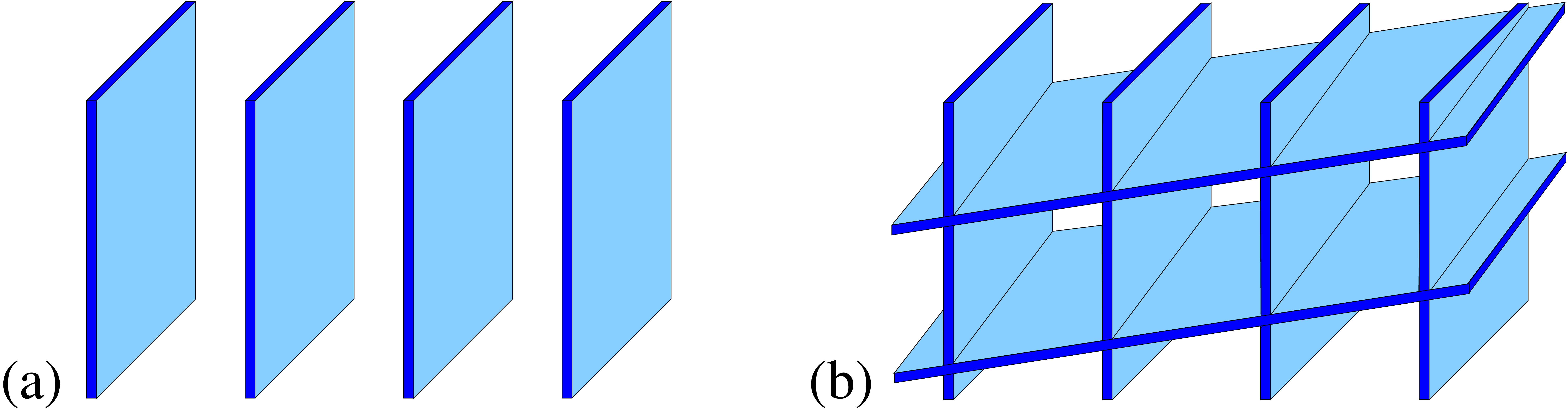}
\caption{Example of walled structures. In (a) we have a ``rank 1'' walled structure and in (b) a ``rank 2'' walled structure. The generalization to walled
structures of any rank is obvious, and precisely defined by the characteristic function \eq{2.14} that is 0 in the walls, and 1 in the remaining material.}
\labfig{2}
\end{figure}

Now define the characteristic 
function
\beq \chi_*(\Bx)=\prod_{k=1}^{p}(1-\Gn_k(\Bx)). \eeq{2.14}
If $p\leq 3$ this is usually a periodic function of $\Bx$ -- an exception being if $p=3$ and there are no nonzero integers $z_1$, $z_2$, and $z_3$
such that $z_1\Bn_1+z_2\Bn_2+z_3\Bn_3=0$. More generally, $\chi_*(\Bx)$ is a quasiperiodic function of $\Bx$. 
The walled structure is where $\chi_*(\Bx)$ takes the value $0$. In the case $p=2$ the wall structure is illustrated in \fig{2}(b).

Recall that a $p$-mode material is a material for which there are $p$ independent strains to which the
material is easily compliant, yet which are much more resistant to any strain in the $(6-p)$-dimensional 
orthogonal subspace. In this sense the microstructure of \fig{0} is a pentamode material. We consider a subclass of multimode materials
which can still support stresses in the limit $\Gd\to 0$. We say a composite with effective tensor $\BC_*$ built from the two materials $\BC_1$ and $\BC_2=\Gd\BC_0$
is easily compliant to a strain $\BGe^0_i$ if the elastic energy $\BGe^0_i:\BC_*\BGe^0_i$ goes to zero as
$\Gd\to 0$, and supports a stress $\BGs^0_j$ if the complementary energy $\BGs^0_j:\BC_*^{-1}\BGs^0_j$ has a nonzero limit as $\Gd\to 0$.
We desire $p$-mode materials for which there are $p$ independent strains to which the
material is easily compliant, yet which supports any stress in the $(6-p)$-dimensional 
orthogonal subspace. The pentamode structure of \fig{0} needs to be modified as all its elastic moduli go to zero as $\Gd\to 0$. The multimode structures we
will introduce have structure on multiple length scales and it is important that one takes the limit of an infinite separation of length scales 
(so one can apply homogenization theory) before taking the limit $\Gd\to 0$.

Inside the walled structure, where $\chi_*(\Bx)=0$ we put a $p$--mode material with effective tensor $\BC_*^{2}=\BC_*(\CV_p)$
that supports any applied stress $\BGs^0$ in the space orthogonal to $\CV_p$ and which is easily
compliant to any strain $\BGe^0$ in the space $\CV_p$. When we take the 6 matrices
\beq \Bv_1=\BGs_1^0/|\BGs_1^0|,\ldots,\Bv_{6-p}=\BGs_{6-p}^0/|\BGs_{6-p}^0|, \Bv_{7-p}=\BGe^0_1/|\BGe^0_1|,\ldots,\Bv_6=\BGe^0_p/|\BGe^0_p| \eeq{2.15}
as an orthonormal basis for the space of 6$\times$6 matrices, we need to find a $p$--mode material for which the elasticity tensor $\BC_*^{2}$ in this basis 
is such that 
\beq \lim_{\Gd\to 0}\BC_*^{2}=\bpm
\BA & 0 \\ 0 & 0
\epm,
\eeq{2.16}
where $\BA$ represents a (strictly) positive definite $(6-p)\times(6-p)$ matrix and the $0$ on the diagonal represents the $p\times p$ zero matrix. 

Outside the walled structure, where $\chi_*(\Bx)=1$ we put the complementary Avellaneda material with effective elasticity tensor
$\BC_*^{1}=\widetilde{\BC}_f^A(\BGs^0_1,\ldots,\BGs^0_{6-p},0,\ldots,0)$. In a variational principle similar to \eq{2.5} (i.e., treating the complementary Avellaneda material and the $p$--mode material both as homogeneous 
materials with effective tensors $\BC_*^{1}$ and $\BC_*^{2}$, respectively) we choose trial stress fields that are constant:
\beq \und{\BGs}_j(\Bx)=\BGs^0_j, \eeq{2.17}
thus trivially fulfilling the differential constraints, and trial strain fields of the form
\beq \und{\BGe}_i(\Bx)=\sum_{k=1}^p\BGe_{i,k}\Gn_k(\Bx)/\Ge, \eeq{2.18}
which are required to have the average values 
\beq \BGe^0_i=\lang \und{\BGe}_i\rang=\sum_{k=1}^p \BGe_{i,k}, \eeq{2.19}
and the matrices $\BGe_{i,k}$ have the form 
\beq  \BGe_{i,k}=a_{i,k}\BGe^{(k)}, \eeq{2.20} 
for some choice of constants $a_{i,k}$
which ensures they are symmetrized rank-1 matrices lying in the space $\CV_p$ (so they cost very little energy), and which ensure that the $\BGe^0_i$ given by
\eq{2.19} are orthogonal. This symmetrized rank-1 form ensures 
 that $\und{\BGe}_i(\Bx)$ derives from a displacement field. Specifically we have
\beq \und{\BGe}_i(\Bx)=[\Grad\und{\Bu}_i(\Bx)+(\Grad\und{\Bu}_i(\Bx))^T]/2, \eeq{2.21}
with
\beq
\und{\Bu}_i(\Bx)=\sum_{k=1}^p a_{i,k}\Ba_k\left\{(\Bn_k\cdot\Bx)\Gn_k(\Bx)/\Ge+([\Bn_k\cdot\Bx]+1)(1-\Gn_k(\Bx))\right\},
\eeq{2.22}
where, as before, $[\Bn_j\cdot\Bx]$ is the greatest integer less than $\Bn_j\cdot\Bx$. One can easily check that this displacement field is continuous at the wall interfaces.

To find upper bounds on the energy associated with this trial strain field, first consider those parts of the wall structure that are outside of any junction regions, i.e., where for some $k$, $\Gn_k(\Bx)=1$, while $\Gn_s(\Bx)=0$ for all $s\ne k$. An upper bound for the volume fraction occupied by the region where $\Gn_k(\Bx)=1$ while $\Gn_s(\Bx)=0$ for all $s\ne k$ is of course $\Ge$ as this represents the 
volume of the region where $\Gn_k(\Bx)=1$. The associated energy per unit volume of the trial strain field in those parts of the wall structure that are outside of any junction regions is bounded above by
\beq \sum_{k=1}^p \BGe_{i,k}:\BC_*(\CV_p)\BGe_{i,k}/\Ge. \eeq{2.22a}
We will see later in Section \sect{multi} that with an appropriate choice of multimode material $\BGe_{i,k}:\BC_*(\CV_p)\BGe_{i,k}$ is bounded above by a quantity proportional to $\Gd$, essentially because all the strain is concentrated in phase 2. So we
require that the limits $\Gd\to 0$ and $\Ge\to 0$ be taken so that $\Gd/\Ge\to 0$ to ensure that the quantity \eq{2.22a} goes to zero in this limit.

Next, consider those junction regions where only two walls meet, i.e., where for some $k_1$ and $k_2> k_1$, $\Bx$  is such that $\Gn_{k_1}(\Bx)=\Gn_{k_2}(\Bx)=1$ 
while $\Gn_s(\Bx)=0$ for all  $s$ not equal to $k_1$ or $k_2$. Provided $\Bn_{k_1}\ne\Bn_{k_2}$, an upper bound for the volume fraction occupied by each such junction region is $\Ge^2$. Then the associated energy per unit volume of the trial strain field in these junction regions where only two walls meet is bounded above by
\beq \sum_{k_1=1}^p \sum_{k_2=k_1+1}^p (\BGe_{i,k_1}+\BGe_{i,k_2}):\BC_*(\CV_p)(\BGe_{i,k_1}+\BGe_{i,k_2}). \eeq{2.22b}
Thus the powers of $\Ge$ cancel and this energy density will go to zero if the multimode material is easily compliant to the strains $\BGe_{i,k_1}+\BGe_{i,k_2}$ for all $k_1$ and $k_2$ with
$k_2>k_1$. 

Finally, consider those junction regions where three or more walls meet, i.e., for some $k_1$, $k_2> k_1$, and $k_3> k_2$, $\Bx$  is such that $\Gn_{k_i}(\Bx)=1$ for $i=1,2,3$. For a given
choice of $k_1$, $k_2> k_1$, and $k_3> k_2$ such that the
three vectors  $\Bn_{k_1}$, $\Bn_{k_2}$, and $\Bn_{k_3}$ are not coplanar an upper bound for the volume fraction occupied by this region is $\Ge^3$. In the case that the three vectors 
$\Bn_{k_1}$, $\Bn_{k_2}$, and $\Bn_{k_3}$ are coplanar, we can ensure that the  volume fraction occupied by this region is  $\Ge^3$ or less by appropriately translating one or two wall structures, i.e., by replacing $\Gn_{k_m}(\Bx)$ with $\Gn_{k_m}(\Bx+\Ga_i\Bn_{k_m})$ for $m=2,3$, for an appropriate choice of $\Ga_2$ and $\Ga_3$ between $0$ and $1$. Since the energy density of the trial field
in these regions scales as $\Ge^3/\Ge^2=\Ge$ we can ignore this contribution in the limit $\Ge\to 0$ as it goes to zero too.

From this analysis of the energy densities associated with the trial fields it follows that one does not necessarily need the pentamode, quadramode, trimode, bimode, and unimode materials as appropriate for the material inside the walled structure. Instead, by modifying the construction, it suffices to use only unimode and bimode materials. In the walled structure we now put unimode materials in those sections where for some $k$, $\Gn_{k}(\Bx)=1$ while $\Gn_{k'}(\Bx)=0$ for 
all $k'\ne k$. Each unimode material is easily compliant to the single strain $\BGe^{(k)}$ appropriate to the wall under consideration.
A prescription for constructing three-dimensional unimode materials, that are multiple rank laminates, and which are easily compliant under any desired single strain is given in Section 5.1 of
Milton and Cherkaev \cite{Milton:1995:WET} .
 In each junction region of the walled structure where $\Gn_{k_1}(\Bx)=\Gn_{k_2}(\Bx)=1$ for some $k_1\ne k_2$ 
while $\Gn_{k}(\Bx)=0$ for all $k$ not equal to $k_1$ or $k_2$, we put a bimode material which is easily compliant to any strain in the subspace spanned by 
$\BGe^{(k_1)}$ and $\BGe^{(k_2)}$ as appropriate to the 
junction region under consideration. At present we do not know of any recipe in three-dimensions for constructing bimode materials that have any desired pair of strains as their easy modes of deformation, other than to superimpose four pentamode structures as described in Section \sect{multi}.
 In the remaining junction regions of the walled
structure (where three or more walls intersect) we put phase 1. The contribution to the average energy of the fields in these regions 
vanishes as $\Ge\to 0$ as discussed above.  

By these constructions we effectively obtain materials with elasticity tensors $\BC_*$ such that
\beq \lim_{\Gd\to 0}\BC_*=(\BI-\GP_p)\widetilde{\BC}_f^A(\BI-\GP_p), \eeq{2.22aa}
where $\BI$ is the fourth-order identity matrix, $\GP_p$ is the fourth-order tensor that is the projection onto the space $\CV_p$, $\BI-\GP_p$ is the projection onto the orthogonal complement of $\CV_p$, and $\widetilde{\BC}_f^A$ is the relevant complementary Avellaneda material.
In the basis \eq{2.15} $\BI-\GP_p$ is represented by the 6$\times$6 matrix that has the block form, 
\beq \BI-\GP_p=\bpm
\BI_{6-p} & 0 \\ 0 & 0
\epm,
\eeq{2.22.0}
where $\BI_{6-p}$ represents the $(6-p)\times(6-p)$ identity matrix and the $0$ on the diagonal represents the $p\times p$ zero matrix. 

\section{Simplifications for  $2$-dimensional printed materials}
\setcounter{equation}{0}

For $2$-dimensional printed materials, or any two-dimensional two-phase composite with one phase being
void, the analysis simplifies as then the space of $2\times 2$ symmetric matrices has dimension $3$ so there are only 4 energy functions
to consider:
\beqa
W_f^0(\BGs^0_1,\BGs^0_2,\BGs^0_3,)
& = &\min_{\BC_*\in GU_f}\sum_{j=1}^3\BGs^0_j:\BC_*^{-1}\BGs^0_j, \nonum
 W_f^1(\BGs^0_1,\BGs^0_2,\BGe^0_1)
& = &\min_{\BC_*\in GU_f}\left[\BGe^0_1:\BC_*\BGe^0_1+\sum_{j=1}^2\BGs^0_j:\BC_*^{-1}\BGs^0_j\right],\nonum
 W_f^2(\BGs^0_1,\BGe^0_1,\BGe^0_2)
& = & \min_{\BC_*\in GU_f}\left[\left(\sum_{i=1}^2\BGe^0_i:\BC_*\BGe^0_i\right)+\BGs^0_1:\BC_*^{-1}\BGs^0_1 \right],\nonum
 W_f^3(\BGe^0_1,\BGe^0_2,\BGe^0_3)
 & = & \min_{\BC_*\in GU_f}\sum_{i=1}^3\BGe^0_i:\BC_*\BGe^0_i.
\eeqa{2.40}
Again $W_f^0(\BGs^0_1,\BGs^0_2,\BGs^0_3)$ is attained for an ``complementary Avellaneda material'' consisting of a sequentially layered laminate geometry
having an effective tensor $\BC_*=\widetilde{\BC}_f^A(\BGs^0_1,\BGs^0_2,\BGs^0_3)\in GU_f$ and we have the inequalities,
\beqa
\sum_{j=1}^2\BGs^0_j:[\widetilde{\BC}_f^A(\BGs^0_1,\BGs^0_2,0)]^{-1}\BGs^0_j & \leq & W_f^1(\BGs^0_1,\BGs^0_2,\BGe^0_1), \nonum
\BGs^0_1:[\widetilde{\BC}_f^A(\BGs^0_1,0,0)]^{-1}\BGs^0_1 & \leq & W_f^2(\BGs^0_1,\BGe^0_1,\BGe^0_2), \nonum
0&\leq & W_f^3(\BGe^0_1,\BGe^0_2,\BGe^0_3),
\eeqa{2.41}
where, as before, the last inequality is sharp in the limit $\Gd\to 0$ being attained when the material consists of islands of phase 1 surrounded by a phase 2.

The recipe for showing that the bound \eq{2.40} on $W_f^1(\BGs^0_1,\BGs^0_2,\BGe^0_1)$ is sharp for certain values of $\BGe^0_1$ and that the bound \eq{2.40} on 
$W_f^2(\BGs^0_1,\BGe^0_1,\BGe^0_2)$ is sharp for certain values of $\BGe^0_1$ and $\BGe^0_2$ is almost exactly the same as in the $3$-dimensional case:
insert into the complementary Avellaneda material a thin walled structure of respectively unimode and bimode materials so that slips can occur along these walls, allowing with
very little energetic cost the average strain $\BGe^0_1$ in the case of $W_f^1$, or any strain in the space spanned by $\BGe^0_1$ and $\BGe^0_2$ in the case of 
$W_f^2$.

\section{The algebraic problem: characterizing those symmetric matrix pencils spanned by symmetrized rank-one matrices}
\labsect{symr1mat}
\setcounter{equation}{0}
We are interested in the following question: \textit{Given $k$ linearly independent symmetric $d\times d$ matrices $\BA_1,\BA_2,\dots,\BA_k$, find necessary and sufficient conditions such that there exist linearly independent matrices $\{\BB_i\}_{i=1}^k$ spanned by the basis elements $\BA_i$ so that each matrix $\BB_i$ is a symmetrized rank one matrix,
i.e., there exist vectors $\Ba_i$ and $\Bb_i$, with $|\Bb_i|=1$ such that $\BB_i=(\Bb_i\Ba_i^T+\Ba_i\Bb_i^T)/2$.
It is assumed that $d=2$ or $3$ and $1\leq k\leq k_d,$ where $k_2=2$ and $k_3=5.$}
Here, we are working in the generic situation, i.e., we prove the algebraic result for a dense set of matrices. The continuity result of Section 9 will allow us to conclude for the whole set of matrices. Actually, the proof below also shows that the algebraic result holds for the complementary of a zero measure set of matrices.

\begin{theorem}
The above problem is solvable if and only if the matrices $\BA_i,$ $i=1,\dots,k$ satisfy the following condition:
\begin{itemize}
\item[(i)]
\beqa
&~& \det(\BA_1)\leq 0,\quad\text{if}\quad k=1,\  d=2, \\&\,&\nonum
&~& \BA_1~ \text{ has two eigenvalues of opposite signs and one zero eigenvalue},\nonum
&~& \quad\quad \text{or has two zero eigenvalues}, \quad\text{if}\quad k=1,\  d=3.
\eeqa{3.3}
\item[(ii)] If $k=d=2$, 
\beqa
\det(\BA_1)<0~\quad\quad\quad\quad\quad &~ &\nonum \text{or}~f(t) = \det(\BA_1+t\BA_2)&~& \text{is quadratic and has two distinct roots for }t,  \nonum
 &~ & \text{or is  linear  in  }t ~\text{with a nonzero coefficient of}~t,~\nonum  &~ &
\eeqa{3.4}
\item[(iii)] If $k=2$ and $d=3$, defining $\BA(\Gn,\Gm)=\Gn\BA_1+\Gm\BA_2$, the numbers
\beq \det(\BA(\Gn,\Gm)),\quad\{\BA(\Gn,\Gm)\}_{11}\{\BA(\Gn,\Gm)\}_{22}-\{\BA(\Gn,\Gm)\}_{12}^2, \quad \{\BA(\Gn,\Gm)\}_{11}
\eeq{3.6}
are never simultaneously non-negative for any choice of $\Gn$ and $\Gm$ not both zero (equivalently $\BA(\Gn,\Gm)$ is never strictly positive definite for any values of $\Gn$ and $\Gm$),
and
\beq
\begin{array}{ll}
\triangle = & 18\det(\BA_1)\det(\BA_2)S_1S_2-4S_1^3\det(\BA_2)+S_1^2S_2^2-4S_2^3\det(\BA_1)
\\*[.2em]
& -\,27\det(\BA_1)^2\det(\BA_2)^2>0,
\end{array}
\eeq{3.7}
where $S_i=\sum_{j=1}^3{s_{ij}},$ $i=1,2$ and $s_{ij}$ is the determinant of the matrix obtained by replacing the $j$-th row of $\BA_i$ by the $j$-th row of $\BA_{i+1},$ where by convention we have $\BA_3=\BA_1$ (equivalently $\BA(\Gn,\Gm)$ has three distinct roots).

\item[(iv)]
\beq
\text{Always solvable if}\quad k\geq 3,\ d=3.
\eeq{3.8}
\end{itemize}
\end{theorem}
\vskip 2mm
\noindent {\bf Remark.}
In fact the condition \eq{3.3} and the last condition in \eq{3.4}, that $f(t)$ is linear in $t$, could be withdrawn since we are considering the generic case. They are inserted because we can treat them explicitly.
\begin{proof}
\textbf{Case (i): $k=1,\ d=2$ or $3.$} In this case $\BA_1$ must be a multiple of $\BB_1$ and hence must be a symmetrized  rank-one matrix. 
To see more clearly the condition for a matrix $\BB$ to be a symmetrized  rank-one matrix, i.e., have the form $\BB=(\Bb\Ba^T+\Ba\Bb^T)/2$ 
let us, without loss of generality, choose our coordinates
so $\Bb=[1,0]^T$ when $d=2$ and $\Bb=[1,0,0]^T$ when $d=3$. Then $\BB$ has the representation
\beq \BB=\bpm
a_1 & a_2/2 \\ a_2/2 & 0
\epm, \quad {\rm when~} d=2, \quad 
\BB=\bpm
a_1 & a_2/2 & a_3/2 \\ a_2/2 & 0 & 0 \\  a_3/2 & 0 & 0
\epm, \quad {\rm when~} d=3.
\eeq{3.9}
These have eigenvalues 
\beqa \Gl  & = & \frac{a_1\pm \sqrt{a_1^2+a_2^2}}{2},\quad {\rm when~} d=2, \nonum
\Gl & = & \frac{a_1\pm \sqrt{a_1^2+a_2^2+a_3^2}}{2},\,\, {\rm and~}\Gl=0, \quad {\rm when~} d=3.
\eeqa{3.10}
So, clearly $\BB$ is a symmetrized rank-one matrix in two-dimensions if and only if $\det(\BB)\leq 0$, and is a symmetrized rank-one matrix in three-dimensions
if and only if it has two eigenvalues of opposite signs and one zero eigenvalue, or has two zero eigenvalues.

\textbf{Case (ii): $k=2,\ d=2.$} In this case there should be two distinct values of $t$ such that $\det(\BA_1+t\BA_2)<0$, which by continuity of this determinant as a function of $t$ is 
guaranteed if any of the conditions in \eq{3.4} are met. Note that the case where $\det(\BA_1+t\BA_2)=0$ for all $t$ can be ruled out from consideration since this can only happen when
$\BA_2$ is proportional to $\BA_1$, as can be easily seen by working in a basis where $\BA_2$ is diagonal.

\textbf{Case (iii): $k=2,\ d=3.$} Consider the matrix pencil (over reals $\Gn$ and $\Gm$) $\BA(\Gn,\Gm)=\Gn \BA_1+\Gm \BA_2$. Assuming that $\det\BA(\Gn,\Gm)$
is not zero for all $\Gn$ and $\Gm$, there are at least two matrices on the pencil which have non-zero determinant. Let us relabel them as $\BA_1$ and $\BA_2$.
Then the equation $\det(\BA(1,\Gm))=0$ must have either two or three roots 
$\Gm=z_i$, $i=1,2$ or $i=1,2,3$, where the $z_i$ are obtained by changing the sign of the generalized eigenvalues. This gives Cardan's condition:
\beq
\begin{array}{ll}
\triangle = & 18\det(\BA_1)\det(\BA_2)S_1S_2-4S_1^3\det(\BA_2)+S_1^2S_2^2-4S_2^3\det(\BA_1)
\\*[.2em]
& -27\,\det(\BA_1)^2\det(\BA_2)^2\geq 0.
\end{array}
\eeq{3.11}
Suppose that $\BA_1+\Gm \BA_2$ contains a symmetric matrix with two zero eigenvalues (a rank-one matrix) as $\Gm$ is varied. Then by redefining $\BA_2$ we can assume $\BA_2$ is this matrix,
now with zero determinant, and by using a basis where $\BA_2$ is diagonal, we see that $\det(\BA_1+\Gm \BA_2)$ depends 
linearly on $\Gm$ and $\det(\BA_1+\Gm \BA_2)$ can only have one root: \eq{3.11} must be violated.
So we can exclude this possibility: $\BA_1+\Gm \BA_2$ has at most one zero eigenvalue for any value of $\Gm$.
Now consider the eigenvalues of $\BA(\Gt)\equiv\BA(\cos\Gt,\sin\Gt)$ as $\Gt$ is varied. As $\BA(-\Gt)=-\BA(\Gt)$ it suffices to consider the interval of $\Gt$ between
0 and $\pi$. Some scenarios for the eigenvalue trajectories are plotted in \fig{3}. At the values $\Gt_i=\arctan^{-1}(z_i)$ at least one of the eigenvalues must be zero,
and the favorable situation is when there are two remaining eigenvalues of opposite signs or only one non-zero eigenvalue. Such angles $\Gt_i$ are marked by the vertical
dashed lines in the figure. The unfavorable situation is when there are two non-zero eigenvalues of the same sign, marked by the red vertical lines in Figure~(a). First suppose that $\BA(\Gt)$ is positive definite for some $\Gt=\Gt_0$. By refining $\theta$ as the old $\theta$ minus $\theta_0$, let us suppose
$\BA(0)$ is positive definite. Then the scenario is that in \fig{3}(a), or some variant of it in which eigenvalues cross, which is unfavorable. 
The only way to avoid this is for $\BA(\Gt)$ to have two zero eigenvalues at the smallest and largest values of $\Gt\in[0,\pi]$ for which $\det\BA(\Gt)=0$,
as in \fig{3}(b), but we have ruled out the possibility that $\BA(\Gt)$ has two zero eigenvalues for any value of $\Gt$.
We are left with \fig{3}(c) as being the only possible suitable scenario. In conclusion, we require that the matrix $\BA(\Gt)$ not be positive semidefinite 
for any choice of $\Gt$, i.e. the three quantities 
\beq \det(\BA(\Gn,\Gm)),\quad\{\BA(\Gn,\Gm)\}_{11}\{\BA(\Gn,\Gm)\}_{22}-\{\BA(\Gn,\Gm)\}_{12}^2, \quad \{\BA(\Gn,\Gm)\}_{11}
\eeq{3.12}
are never simultaneously non-negative for any choice of $\Gn$ and $\Gm$ not both zero. This condition could be made explicit by using the formula for the roots of a cubic
to determine the generalized eigenvalues $-z_i$.

\begin{figure}[!ht]
\centering
\includegraphics[width=0.85\textwidth]{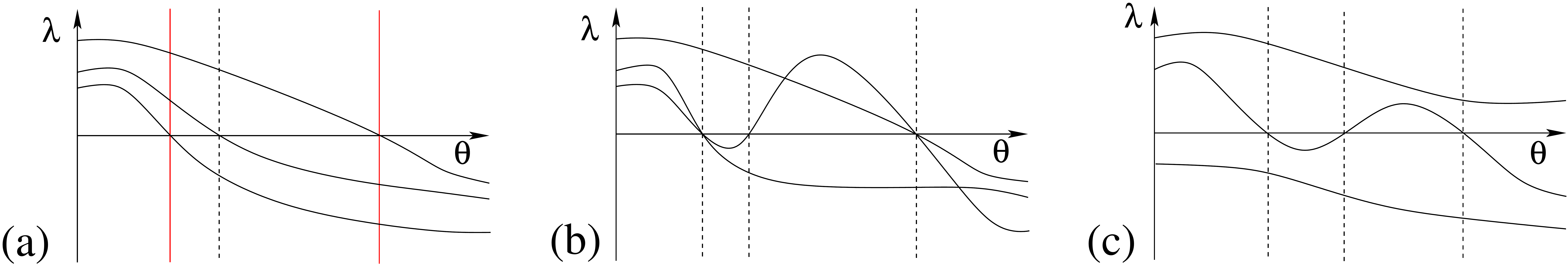}
\caption{Some scenarios for the eigenvalues $\Gl$ of $\BA(\Gt)=\cos\Gt \BA_1+\sin\Gt \BA_2$ as $\Gt$ is varied.}
\labfig{3}
\end{figure}

\textbf{Case (iv): $k\geq 3,\ d=3.$}
The case $k=3$ is a straightforward consequence of Lemma~7.1 below.
\par
It remains to consider $k\geq 4$ and $d=3.$ By the previous step, in the space spanned by $\BA_1$, $\BA_2$, and $\BA_3$ there are three matrices $\BA_1'$, $\BA_2'$ and 
$\BB_3=\BA_3+\Gn_3\BA_1'+\mu_3\BA_2'$ that are linearly independent symmetrized rank-one matrices. Then again by the previous step 
we can find linearly independent matrices $\BB_1,\ldots,\BB_k$ that have the form 
$\BB_1=\BA_1',$ $\BB_2=\BA_2'$ and $\BB_i=\BA_i+\Gn_i\BA_1'+\mu_i\BA_2'$ for $3\leq i\leq k$, that are linearly independent and are symmetrized rank-one matrices.
Thus the proof is finished.
\end{proof}
In the sequel we denote
\beq
\Ba\otimes\Bb:=\Ba\Bb^T\quad\mbox{and}\quad \Ba\odot\Bb:=(\Ba\otimes\Bb+\Bb\otimes\Ba)/2
\quad\mbox{for }\Ba,\Bb\in\R^3.
\eeq{3.13}
\begin{lemma}
Let $\BA,\BB,\BC$ be three symmetric matrices of $\R^{3\times 3}$.
\par\medskip\noindent
$i)$ Up to small perturbations of the matrices $\BA,\BB,\BC$, there exist a basis $(\Bx,\By,\Bz)$ of $\R^3$ and three vectors $\Ba,\Bb,\Bc$ of $\R^3$ satisfying
\beq
\begin{cases}
\Ba\in\{\BA\Bx,\BB\Bx,\BC\Bx\}^\perp\setminus\{\b0\}
\\
\Bb\in\{\BA\By,\BB\By,\BC\By\}^\perp\setminus\{\b0\}
\\
\Bc\in\{\BA\Bz,\BB\Bz,\BC\Bz\}^\perp\setminus\{\b0\},
\end{cases}
\eeq{3.14}
or equivalently,
\beq
\Ba\odot \Bx,\,\Bb\odot \By,\,\Bc\odot \Bz\in\{\BA,\BB,\BC\}^\perp\setminus\{\b0\}.
\eeq{3.15}
\par\medskip\noindent
$ii)$ Up to small perturbations of $\BA,\BB,\BC$, there exist three independent symmetrized rank-one matrices in the space $\{\BA,\BB,\BC\}^\perp$.
\end{lemma}
\begin{proof}
$i)$ Let $F$ be the cubic function defined by
\beq
F(\Bx):=\det\left(\BA\Bx,\BB\Bx,\BC\Bx\right)\quad\mbox{for }\Bx\in\R^3.
\eeq{3.16}
If $F\equiv 0$ in $\R^3$, then condition \eq{3.14} is immediately satisfied.
Otherwise, there exists a basis $(\Bx_0,\Bu_0,\Bv_0)$ of $\R^3$ in the non-empty open set $\{F\neq 0\}$.
Since we have
\beq
F(\Bx_0+s\Bu_0)\mathop{\sim}_{|s|\to\infty}s^3 F(\Bu_0)\mathop{\longrightarrow}_{|s|\to\infty}\pm\infty,
\eeq{3.17}
there exists $s,t\in\R\setminus\{0\}$ such that $\Bx:=\Bx_0+s\Bu_0$ and $\By:=\Bx_0+t\Bv_0$ are two independent vectors in the set $\{F=0\}$.
\par
First, assume that the set $\{F=0\}$ is not contained in the plane ${\rm Span}\left\{\Bx,\By\right\}$.
Then, there exists a basis $(\Bx,\By,\Bz)$ of $\R^3$ in the set $\{F=0\}$. Therefore, there exist three vectors $\Ba,\Bb,\Bc$ of~$\R^3$ satisfying~\eq{3.14}, or equivalently \eq{3.15}.
\par
Now, assume that $\{F=0\}\subset{\rm Span}\left\{\Bx,\By\right\}$. First of all, up to small perturbations we can assume that the matrices $\BA,\BB,\BC$ are invertible. Since $\BB^{-1}\BC$ is a $(3\times 3)$ real matrix, it has at least a real eigenvalue $\Gl$.
The perturbation procedure is now divided in two cases.
\par\medskip\noindent
{\it First case:} The matrix $\BB^{-1}\BC$ has two complex conjugate eigenvalues.
\par\smallskip\noindent
Then, the eigenspace ${\rm Ker}\left(\BB^{-1}\BC-\Gl\,\BI_3\right)$ is a line of $\R^3$ spanned by $\Be\in\R^3\setminus\{\b0\}$.
Consider a basis $(\Bx_0,\Bu_0,\Bv_0)$ of $\R^3$ in the set $\{F\neq 0\}$ such that $(\Be,\Bx_0,\Bu_0)$ and $(\Be,\Bx_0,\Bv_0)$ are also two bases of $\R^3$. As previously there exist $s,t\in\R\setminus\{0\}$ such that $\Bx:=\Bx_0+s\Bu_0$ and $\By:=\Bx_0+t\Bv_0$ are two independent vectors of the set $\{F=0\}$. Moreover, since $(\Be,\Bx)$ and $(\Be,\By)$ are two families of independent vectors and $\R\,\Be$ is the unique real eigenspace of the matrix $\BB^{-1}\BC$, we have
\beq
\BB\Bx\times \BC\Bx\neq\b0\quad\mbox{and}\quad\BB\By\times\BC\By\neq\b0.
\eeq{3.18}
Now, consider a vector $\Bu\in\{\Bx,\By\}^\perp\setminus\{\b0\}$, and the matrix $\BM\in\R^{3\times 3}$ defined by
\beq
\BM\Bx=\BGx,\quad \BM\By=\BGn,\quad\BM\Bu=\b0,
\eeq{3.19}
where the vectors $\BGx,\BGn$ will be chosen later.
Define for $\tau>0$, the perturbed function
\beq
F_\tau(\Bz):=\det\left(\BA\Bz+\tau\,\BM\Bz,\BB\Bz,\BC\Bz\right)\quad\mbox{for }\Bz\in\R^3.
\eeq{3.20}
We have
\beq
\begin{cases}
F_\tau(\Bx+\tau\Bu)=\tau\,\BGx\cdot\big(\BB\Bx\times \BC\Bx+\BO(\tau)\big)+O(\tau)
\\
F_\tau(\By+\tau\Bu)=\tau\,\BGn\cdot\big(\BB\By\times \BC\By+\BO(\tau)\big)+O(\tau),
\end{cases}
\eeq{3.21}
where $\BO(\tau)$ denote some first-order vectors in $\tau$ and $O(\tau)$ some first-order real numbers in $\tau$ which are independent of $\BGx,\BGn$.
Condition \ref{3.18} then allows us to choose $\BGx=\BGx_\tau$, $\BGn=\BGn_\tau$ such that $F_\tau(\Bx+\tau\Bu)=F_\tau(\By+\tau\Bu)=0$.
Therefore, since $(\Bx,\By,\Bu)$ is a basis of~$\R^3$, $(\Bx,\Bx+\tau\Bu,\By+\tau\Bu)$ is also a basis of~$\R^3$ which in addition lies in the set $\{F_\tau=0\}$. This leads us to condition \eq{3.14} with the matrices $\BA+\tau\,\BM,\BB,\BC$.
\par\medskip\noindent
{\it Second case:} The matrix $\BB^{-1}\BC$ has only real eigenvalues.
\par\smallskip\noindent
Then, there exists a small perturbation $\BC_\tau$ of $\BC$ such that the perturbed matrix $\BB^{-1}\BC_\tau$ has $3$ distinct real eigenvalues. Hence, the matrix $\BB^{-1}\BC_\tau$ admits a basis $(\Bx,\By,\Bz)$ of eigenvectors, which implies that
\beq
\BC_\tau\,\Bx-\Gl\,\BB\Bx=\BC_\tau\,\By-\Gl\,\BB\By=\BC_\tau\,\Bz-\Gl\,\BB\Bz=\b0.
\eeq{3.22}
Therefore, the perturbed function
\beq
F_\tau(\Bu):=\det\left(\BA\Bu,\BB\Bu,\BC_\tau\Bu\right)\quad\mbox{for }\Bu\in\R^3,
\eeq{3.23}
satisfies $F_\tau(\Bx)=F_\tau(\By)=F_\tau(\Bz)=0$, which again leads us to condition \eq{3.14} with the matrices $\BA,\BB,\BC_\tau$.
\par\bigskip\noindent
$ii)$ We will distinguish four cases according to the following conditions satisfied or not by the basis $(\Bx,\By,\Bz)$ of $\R^3$ and the vectors $\Ba,\Bb,\Bc\in\R^3\setminus\{0\}$ obtained in step~$i)$:
\beq
\begin{cases}
\Ba\in{\rm Span}\left\{\Bx,\By\right\}\cap{\rm Span}\left\{\Bx,\Bz\right\}
\\
\Bb\in{\rm Span}\left\{\By,\Bx\right\}\cap{\rm Span}\left\{\By,\Bz\right\}
\\
\Bc\in{\rm Span}\left\{\Bz,\Bx\right\}\cap{\rm Span}\left\{\Bz,\By\right\}.
\end{cases}
\eeq{3.24}
{\it First case:} The three vectors $\Ba,\Bb,\Bc$ satisfy conditions \eq{3.24}.
\par\smallskip\noindent
Then, since $(\Bx,\By,\Bz)$ is a basis of $\R^3$, we have necessarily $\Ba\in\R\,\Bx$, $\Bb\in\R\,\By$, $\Bc\in\R\,\Bz$.
Therefore, $\Bx\odot\Bx,\By\odot\By,\Bz\odot\Bz$ are clearly three independent matrices of $\{\BA,\BB,\BC\}^\perp$.
\par\medskip\noindent
{\it Second case:} $\Bb,\Bc$ satisfy conditions \eq{3.24} but not $\Ba$.
\par\smallskip\noindent
Then, for example $(\Ba,\Bx,\By)$ is a basis of $\R^3$, and $\Bb\in\R\,\By$, $\Bc\in\R\,\Bz$.
Let $\Bu\in\{\By,\Bz\}^\perp\setminus\{\b0\}$, and let $\Ga,\beta,\Gg$ such that $\Ga\,\Ba\odot\Bx+\beta\,\By\odot\By+\Gg\,\Bz\odot\Bz=\b0$. Multiplying by $\Bu$ we get that $\Ga(\Bx\cdot\Bu)\,\Ba+\Ga(\Ba\cdot\Bu)\,\Bx=\b0$, hence $\Ga=0$ since $\Bx\cdot\Bu\neq 0$. We deduce immediately that $\beta=\Gg=0$. Therefore, $\Ba\odot\Bx,\By\odot\By,\Bz\odot\Bz$ are three independent matrices of $\{\BA,\BB,\BC\}^\perp$.
\par\medskip\noindent
{\it Third case:} $\Ba,\Bb$ do not satisfy conditions \eq{3.24}, with $\Ba\notin{\rm Span}\left\{\Bx,\By\right\}$  and $\Bb\notin\R\,\Ba\cup\R\,\Bx$ (respectively $\Ba\notin{\rm Span}\left\{\Bx,\Bz\right\}$ and $\Bc\notin\R\,\Ba\cup\R\,\Bx$).
\par\smallskip\noindent
Then, $(\Ba,\Bx,\By)$ is a basis of $\R^3$. Let $\Bu\in\{\Bx,\By\}^\perp\setminus\{\b0\}$, and let $\Ga,\beta\in\R$ such that $\Ga\,\Ba\odot\Bx+\beta\,\Bb\odot\By=\b0$. Multiplying by $\Bu$ we get that $\Ga(\Ba\cdot\Bu)\,\Bx+\beta(\Bb\cdot\Bu)\,\By=\b0$, hence $\Ga=0$ since $\Ba\cdot\Bu\neq 0$, and thus $\beta=0$. Therefore, $\Ba\odot\Bx,\Bb\odot\By$ are two independent matrices of $\{\BA,\BB,\BC\}^\perp$, which have two eigenvalues of opposite signs and one $0$ eigenvalue.
\par
Let us prove by contradiction that
\beq
\exists\,t\in\R\setminus\{0\},\quad\det\left(\Ba\odot\Bx+t\,\Bb\odot\By\right)\neq 0.
\eeq{3.25}
Otherwise, for any $t\neq 0$, there exits $\Bz_t\in{\rm Ker}\left(\Ba\odot\Bx+t\,\Bb\odot\By\right)\setminus\{\b0\}$, hence
\beq
(\Bx\cdot\Bz_t)\,\Ba+(\Ba\cdot\Bz_t)\,\Bx+t(\By\cdot\Bz_t)\,\Bb+t(\Bb\cdot\Bz_t)\,\By=\b0.
\eeq
Since $(\Ba,\Bx,\By)$ is a basis of $\R^3$ and $\Bz_t\neq\b0$, we have necessarily $\By\cdot\Bz_t\neq 0$, which implies that
\beq
-\,\Bb={\Bx\cdot\Bz_t\over t(\By\cdot\Bz_t)}\,\Ba+{\Ba\cdot\Bz_t\over t(\By\cdot\Bz_t)}\,\Bx+{\Bb\cdot\Bz_t\over\By\cdot\Bz_t}\,\By
=\Ga\,\Ba+\beta\,\Bx+\Gg\,\By,
\eeq{3.26}
where $\Ga,\beta,\Gg$ are independent of $t$, and
\beq
(\Bx-\Ga t\,\By)\cdot\Bz_t=(\Ba-\beta t\,\By)\cdot\Bz_t=(\Bb-\Gg\,\By)\cdot\Bz_t=0.
\eeq{3.27}
Since $\Bz_t\neq\b0$ there exists $(p_t,q_t,r_t)\in\R^3\setminus\{0\}$ such that
\beq
\begin{array}{l}
p_t(\Bx-\Ga t\,\By)+q_t(\Ba-\beta t\,\By)+r_t(\Bb-\Gg\,\By)
\\*[.2em]
=(q_t-\Ga r_t)\,\Ba+(p_t-\beta r_t)\,\Bx-(\Ga t p_t+\beta t q_t+2\Gg r_t)\,\By=\b0,
\end{array}
\eeq{3.28}
which implies that $q_t=\Ga r_t$, $p_t=\beta r_t$ and $r_t(\Ga\beta t+\Gg)=0$. Since $(p_t,q_t,r_t)\neq\b0$, we have $r_t\neq 0$ and $\Ga\beta t+\Gg=0$ for any $t\neq 0$, hence $\Ga\beta=0$ and $\Gg=0$. This yields a contradiction between~\eq{3.26} and $\Bb\notin\R\,\Ba\cup\R\,\Bx$.
\par
By virtue of \eq{3.25} there exist two non-zero real numbers $\Ga\neq\beta$ such that the matrices
\beq
\BM:=\Ba\odot\Bx+\Ga\,\Bb\odot\By\quad\mbox{and}\quad \BN:=\Ba\odot\Bx+\beta\,\Bb\odot\By\quad\mbox{are invertible.}
\eeq
The function $p(t):=\det\left(\beta\,\BM-t\,\BN\right)$ is a polynomial of degree $3$ whose $\Ga,\beta$ are two distinct roots.
Then, the polynomial $p(t)$ must change sign by crossing $\Ga$ for example (the conclusion is similar for~$\beta$).
Let $\Gl_1(t)\leq\Gl_2(t)\leq\Gl_3(t)$ be the well-ordered eigenvalues of the symmetric matrix $\beta\,\BM-t\,\BN$.
Since the vectors $\Ba,\Bx$ are independent, $\Ba\odot\Bx$ has two eigenvalues of opposite signs and one $0$ eigenvalue, hence
$\Gl_1(\Ga)<\Gl_2(\Ga)=0<\Gl_3(\Ga)$.
\par
Now, let $\BP_\tau$ for a small $\tau>0$, be a symmetric matrix in the space $\{\BA,\BB,\BC\}^\perp$,
such that $|\BP_\tau-\Ba\odot\Bx|=O(\tau)$, and such that the three matrices $\Ba\odot\Bx,\Bb\odot\By,\BP_\tau$ are independent (note that the dimension of $\{\BA,\BB,\BC\}^\perp$ is $\geq 3$).
Define the two perturbed matrices
\beq
\BM_\tau:=\BP_\tau+\Ga\,\Bb\odot\By\quad\mbox{and}\quad \BN_\tau:=\BP_\tau+\beta\,\Bb\odot\By.
\eeq{3.29}
Since the well-ordered eigenvalues of a real symmetric matrix $\BS$ are Lipschitz-continuous with respect to $\BS$ (see, {\em e.g.}, \cite{Ciarlet:1989:INL}, Theorem~2.3-2), the eigenvalues $\Gl_1^\tau(t)\leq\Gl_2^\tau(t)\leq\Gl_3^\tau(t)$ of $\beta\,\BM_\tau-t\,\BN_\tau$ converge uniformly as $\tau\to 0$ to the eigenvalues $\Gl_1(t)\leq\Gl_2(t)\leq\Gl_3(t)$ of $\beta\,\BM-t\,\BN$, with respect to $t$ in a neighborhood of $\Ga$.
Hence, for $\tau>0$ small enough, there exist $\Ga_\tau$ close to $\Ga$ such that $\Ga_\tau\neq\beta$ and
$\Gl_1^\tau(\Ga_\tau)<\Gl_2^\tau(\Ga_\tau)=0<\Gl_3^\tau(\Ga_\tau)$.
Then, by \ref{3.29} there exists $\Bc_\tau,\Bz_\tau\in\R^3$ such that
\beq
\beta\,\BM_\tau-\Ga_\tau\,\BN_\tau=\Bc_\tau\odot\Bz_\tau=(\beta-\Ga_\tau)\,\BP_\tau+\beta(\Ga-\Ga_\tau)\,\Bb\odot\By,
\quad\mbox{with }\beta-\Ga_\tau\neq 0.
\eeq{3.30}
Therefore, $\Ba\odot\Bx,\Bb\odot\By,\Bc_\tau\odot\Bz_\tau$ are three independent symmetrized rank-one matrices in the space $\{\BA,\BB,\BC\}^\perp$.
\par\medskip\noindent
{\it Fourth case:} $\Ba,\Bb$ do not satisfy condition \eq{3.24}, with $\Ba\notin{\rm Span}\left\{\Bx,\By\right\}$ and $\Bb\in\R\,\Ba\cup\R\,\Bx$ (respectively $\Ba\notin{\rm Span}\left\{\Bx,\Bz\right\}$ and $\Bc\in\R\,\Ba\cup\R\,\Bx$).
\par\smallskip\noindent
For example we have $\Bb\in\R\,\Ba$. We thus start from the matrices $\Ba\odot\Bx,\Ba\odot\By$ in the space $\{\BA,\BB,\BC\}^\perp$, where $(\Ba,\Bx,\By)$ is a basis of $\R^3$. We will consider a perturbation of $\BA,\BB,\BC$ for leading us to the third case.
\par
Let $\Bt\in\{\Ba,\Bx\}^\perp\setminus\{\b0\}$, let $\Bd\in\R^3\setminus\big(\R\,\Ba+\R\,\Bx\big)$, and consider for a small $\tau>0$, the perturbed vector $\Bb_\tau:=\Ba+\tau\,\Bd\notin\R\,\Ba\cup\R\,\Bx$, and the perturbed matrices
\beq
\BA_\tau:=\BA+\tau\,\Bt\odot\Bu_\tau,\quad \BB_\tau:=\BB+\tau\,\Bt\odot\Bv_\tau,\quad \BC_\tau:=\BC+\tau\,\Bt\odot\Bw_\tau,
\eeq
where the vectors $\Bu_\tau,\Bv_\tau,\Bw_\tau$ will be chosen later.
We have clearly $\Ba\odot\Bx\in \{\BA_\tau,\BB_\tau,\BC_\tau\}^\perp$.
On the other hand, we have
\beq
\BA_\tau:\Bb_\tau\odot\By=\tau\,\big(\BA:\Bd\odot\By+\Bt\odot\Bu_\tau:\Ba\odot\By+\tau\,\Bt\odot\Bu_\tau:\Bd\odot\By\big).
\eeq{3.31}
Since $2\,\Bt\odot\Bu_\tau:\Ba\odot\By=(\Bt\cdot\By)\,\Ba\cdot\Bu_\tau$ with $\Bt\cdot\By\neq 0$, we can choose $\Bu_\tau=\BO(1)$ with respect to $\tau$, such that $\BA_\tau:\Bb_\tau\odot\By=0$. Hence, choosing similarly $\Bv_\tau,\Bw_\tau$, we get that $\Bb_\tau\odot\By\in \{\BA_\tau,\BB_\tau,\BC_\tau\}^\perp$.
Therefore, the vectors $\Ba,\Bb_\tau,\Bx,\By$ satisfy the conditions of the third case with the perturbed matrices $\BA_\tau,\BB_\tau,\BC_\tau$.
\end{proof}

\section{Constructing suitable multimode materials for the wall microstructure}
\labsect{multi}
\setcounter{equation}{0}

Let us specify the construction of the desired multimode materials in 2 dimensions and then move to three dimensions. We begin by constructing 
bimode materials that can only support one stress. One could use the fourth-rank laminate structure described in detail in section 30.7 of
\cite{Milton:2002:TOC}. The analysis would then be essentially a repeat of that analysis which builds the appropriate trial stress and strain
fields at each length scale. The key feature is that these trial fields need to be chosen so the trial stress associated with the average stress 
 $\BGs^0$ we want to achieve at the macroscopic scale is concentrated entirely in phase 1 (apart from boundary layers that we ignore, whose contribution to the
energy vanishes in the homogenization limit), and so the trial strain associated with an average strain that is orthogonal to $\BGs^0$ is concentrated entirely in phase 2. 

Rather than doing this, it is more instructive to build trial stress and strain fields that are concentrated in phase 1 and phase 2 respectively 
for the honeycomb and inverted honeycomb bimode structures of \fig{4}, as the ideas here carry over to pentamode materials.
The trial stress is easy. It is taken to be macroscopically constant with a value $\Ga_i\Ba_i\Ba_i^T$ in each strut which is parallel to the unit vector $\Ba_i$
in \fig{5}.
Let $w_i$ denote the width of the strut parallel to $\Ba_i$, for $i=0,1,2$. Since the net ``force'' on the black junction regions in \fig{5}(a) or \fig{5}(b) must be zero we obtain
\beq 0=-\sum_{k=0}^2w_i(\Ga_i\Ba_i\Ba_i^T)\Ba_i=-\sum_{k=0}^2w_i\Ga_i\Ba_i. \eeq{4.1}
Since $w_1=w_2$ and $\Ba_0$ points in the horizontal direction, while $\Ba_1$ and $\Ba_2$ have the same horizontal component and equal but opposite vertical components we get
\beq \Ga_1=\Ga_2=-w_0\Ga_0/[2w_1(\Ba_1\cdot\Ba_2)]. \eeq{4.2}
The symmetry of the trial stress field implies there is no associated torque acting on the junction regions. The trial stress in the junction regions is really not that important.
One choice is the stress field that satisfies the elasticity equations appropriate to phase 1 filling the junction region when constant tractions 
act on the three sides. The average value of the trial stress does not depend on the choice of trial stress in the junctions. Indeed, since $\Div(\BGs)=0$ it follows by
integration by parts of $\Div(\BGs\Bx)$  (where $\BGs\Bx$ is a third order tensor) that 
\beq \int_{\GO}\BGs\,d\Bx=\int_{\Md\GO}\Bt\Bx^T\,dS,\quad\quad{\rm where}~\Bt=\BGs\Bn~{\rm is~the~surface~traction},
\eeq{4.2a}
in which $\GO$ is any region with boundary $\Md\GO$. For example, the boundary of $\GO$ could be the outermost boundary of the shapes in \fig{5}(a) and \fig{5}(b)
respectively, where we include the dashed lines as part of the boundary.

In passing, we remark that if $\BGs^0$ is proportional to the identity matrix, then the microstructure of 
\fig{5}(a) resembles a Sigmund microstructure (see the last subfigure in figure 2 in \cite{Sigmund:2000:NCE}).
However, we do not require the tuning of layer widths in the struts that makes his structure optimal. 
Suboptimal structures are perfectly fine in the walls, since the walls ultimately occupy a vanishingly small volume fraction
in the final material.

To obtain a trial easy strain it suffices to specify the trial displacement in the unit cell. We only choose motions so the junction regions
(triangular in \fig{5}(c) and quadrilateral in \fig{5}(d)) undergo rigid body translations so there is no strain inside them.
Thus associated with \fig{5}(c) one can clearly identify two
independent macroscopic modes of motion. The first is where the line RS moves vertically upwards while the line PU remains fixed, and Q and T move in such a way that the 
lengths QR, QP, TS, and TU remain equal and preserved in length. One can choose the displacement to be linear in each of the three regions $A$, $B$, and $C$ so that it matches the
displacement on the boundary. The second is where the line  RS moves horizontally while the line PU remains fixed, and Q and T move in such a way that the 
lengths QR, QP, TS, and TU remain equal and preserved in length. In either case inside the horizontal laminate arm there is no strain, while inside the inclined 
laminate arms there is an infinitesimal shear so the junction at $P$ remains fixed, while the junction at $Q$ moves perpendicular to $\Ba_1$ while
the junction at $V$ moves perpendicular to $\Ba_2$. We also note that there is also an easy microscopic motion which results in no macroscopic motion. Define
the center of each triangular junction to be the point which is at the junction of the perpendicular bisector of the three faces. Then if all the 
triangular junctions undergo the same infinitesimal rotation about these centers while the laminate material in the struts shears at the same time,
it will cost very little energy. The trial strain field is bounded and non-zero only in phase 2, and therefore the associated upper bound on the elastic energy
scales in proportion to $\Gd$. 

The situation in \fig{5}(d) is basically similar. The two black quadrilateral junction regions at the bottom of the figure can remain fixed. Then one mode is the symmetric one, where the region $A$ undergoes uniaxial compression in the horizontal direction and at the same time moves downwards. The second is where the region $A$ undergoes pure shear, so the junction on the left side of it moves up, while the right side moves down. The strain field can be taken constant in the regions A,B,C,D, and E, and in the inclined laminate strut arms is also constant and corresponds to pure shear. These strains are easily determined from the value of the trial displacement field at the boundaries of each region. Again, the trial strain field is bounded and non-zero only in phase 2, and therefore the associated upper bound on the elastic energy
scales in proportion to $\Gd$.

\begin{figure}[!ht]
\centering
\includegraphics[width=0.7\textwidth]{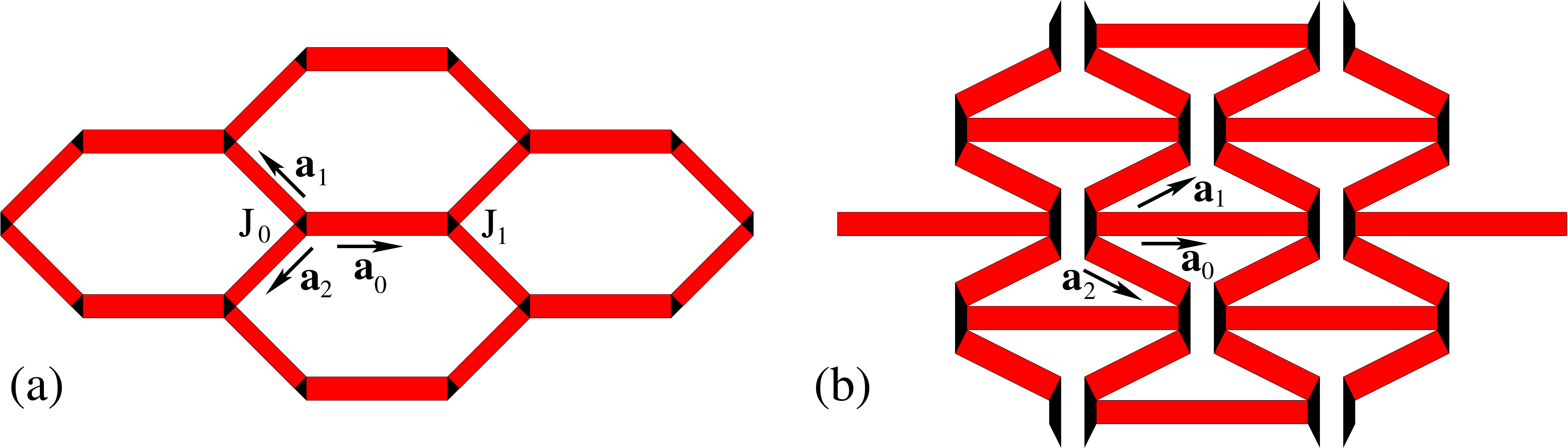}
\caption{Two-dimensional bimode materials that can only support one average stress field $\BGs^0$, and which are easily compliant to any strain orthogonal to $\BGs^0$. Here the
red struts are laminates of the two phases with the interfaces in the laminate parallel to the direction of the struts.
 The geometry of (a) is appropriate if $\det\BGs^0>0$, the geometry of (b) is appropriate if $\det\BGs^0<0$, and if $\det\BGs^0=0$ it suffices to use a simple laminate with the layer surfaces perpendicular to the null vector of $\BGs^0$.}
\labfig{4}
 \end{figure}

\begin{figure}[!ht]
\centering
\includegraphics[width=0.7\textwidth]{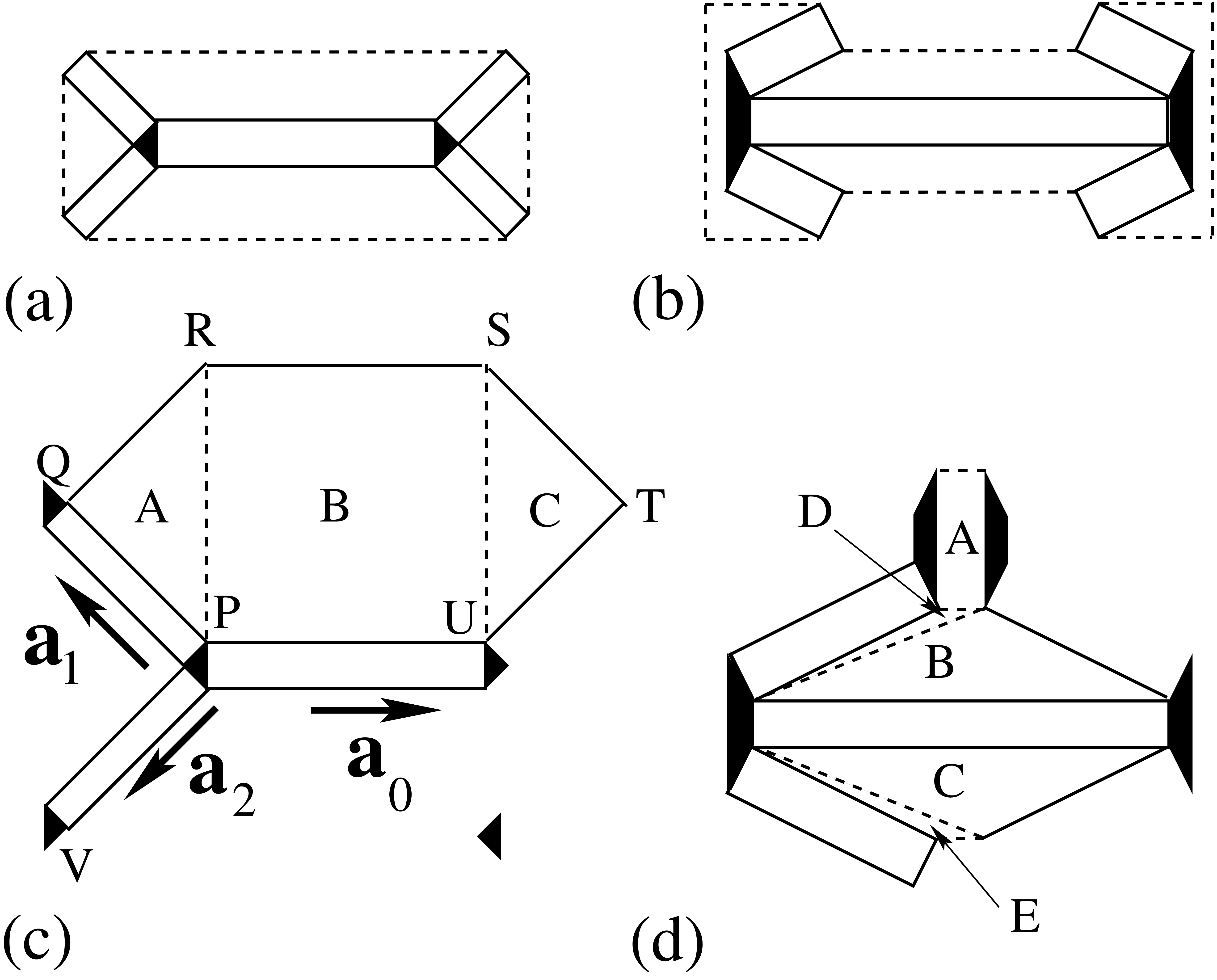}
\caption{The honeycomb structure of \fig{4}(a) can be taken to have the unit cell shown in figure (a) of this drawing. Similarly the inverted honeycomb structure of \fig{4}(a) 
can be taken to have the unit cell shown in figure (b) of this drawing. The space outside the struts and junction regions (that is occupied by phase 2) has been triangulated
with boundaries marked by the dashed lines to make the construction of the trial stress fields easy.}
\labfig{5}
\end{figure}

The structures of \fig{6} give suitable two-dimensional unimode materials. We will not specific the appropriate trial stress and strain fields which prove that these structures 
have the desired elastic behavior as they are exactly the same as given in section 30.6 of Milton \cite{Milton:2002:TOC}.  

\begin{figure}[!ht]
\centering
\includegraphics[width=0.9\textwidth]{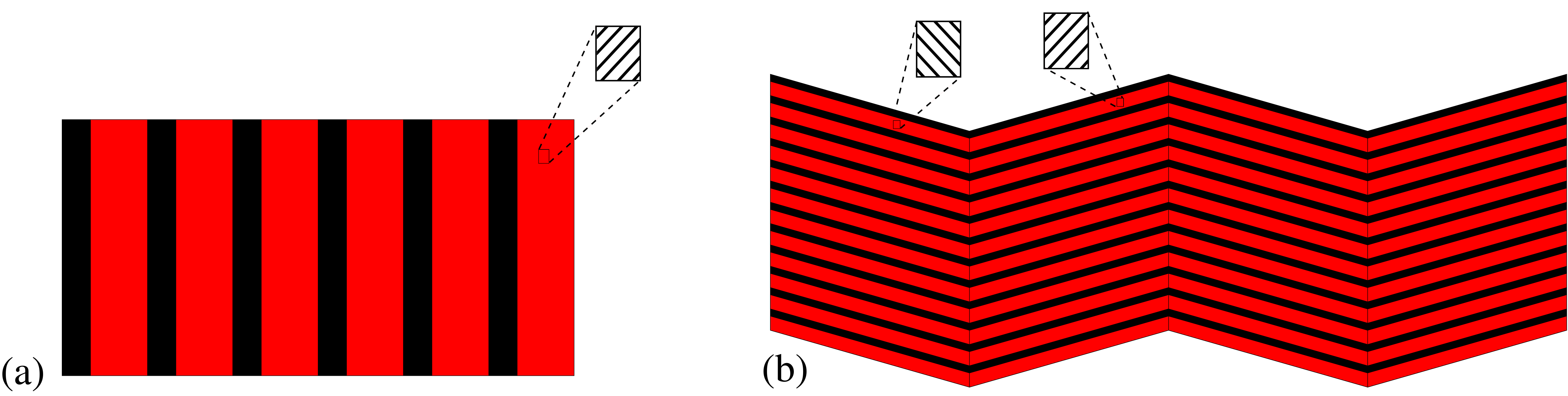}
\caption{Two-dimensional unimode materials that is easily compliant to one average strain field $\BGe^0$, and which can support any stress orthogonal to $\BGe^0$. In both
(a) and (b) the red region represents a laminate as indicated by the inserts.
The second-rank laminate geometry of (a) is
appropriate if $\det\BGe^0<0$ and the third-rank geometry of (b) is appropriate for any $\BGe^0$.}
\labfig{6}
\end{figure}

We now describe the pentamodes and the trial fields in them. Given four vectors $\Ba_0,\Ba_1,\Ba_2$, and $\Ba_4$ (no longer required to be unit vectors) we position a point $P$ at the origin, and join $P$
to the four points $\Bx=\Ba_i$, $i=0,1,2,3$, with four infinitesimally thin rods, as in \fig{7}(a). We then take as our unit cell of periodicity the parallelepiped with the eight points
$\Bx=\Ba_i$, $\Bx=\Ba_1+\Ba_2+\Ba_3-\Ba_i-\Ba_0$, $i=0,1,2,3$ (the three vectors $\Bv_i=\Ba_i-\Ba_0$, $i=1,2,3$ are the primitive lattice vectors). We require that $\Ba_0,\Ba_1,\Ba_2,\Ba_4$
be chosen so $P$ lies within this parallelepiped. After periodically extending the rod structure (with rods joining $k_1\Bv_1+k_2\Bv_1+k_3\Bv_1$ with the four points
$k_1\Bv_1+k_2\Bv_1+k_3\Bv_1+a_i$, $i=0,1,2,3$, for any integers $k_1$, $k_2$, and $k_3$, we then coat this periodic rod structure with phase 1, as illustrated in \fig{7}(b)
so that any point $\Bx$ is in phase 1 if and only if it 
is within a distance $r$ of the rod structure. 
Here $r$ should be chosen appropriately small so that the coatings of each rod contain a cylindrical section that we refer to as a strut. \fig{7}(b) is misleading as it suggests
that the unit cell only contains one junction region. The true structure which should be periodically repeated (by making copies shifted by vectors $k_1\Bv_1+k_2\Bv_1+k_3\Bv_1$
for all combinations of integers $k_1$, $k_2$, and $k_3$) is shown in \fig{7}(c), and contains the junction of \fig{7}(b) plus the one obtained by inverting it 
under the transformation $\Bx\to -\Bx$. The final
step, illustrated in \fig{7}(c), is to take a cylindrical subsection of each cylindrical section between junctions, and replace it with a pentamode material that
supports any stress proportional to $\Ba_i\Ba_i^T$. It is convenient to take end faces of the cylindrical subsection to be perpendicular to the cylinder axis, i.e. perpendicular
to the vector $\Ba_i$ that is parallel to the cylinder axis. Now we define the junction regions to be those connected regions of phase 1 that are bounded by the cylindrical
subsections. 

To obtain the trial stress field, we first solve for the tensions in the rods of \fig{7}(a) when the rods are completely rigid and supporting a stress. These are found just by
balance of forces at the junctions. If the rods parallel to $\Ba_i$ have a tension $T_i$ (that could be negative) then we take in the cylindrical subsection of the 
corresponding strut of the final pentamode a trial stress field $T_i\Ba_i\Ba_i^T/|\Ba_i|^2/(\pi r^2)$ giving rise to a net force $T_i$ pulling (pushing if $T_i$ is negative)
on the adjacent junction regions.
Inside the junction region we take a stress field that satisfies the elasticity equations appropriate to phase 1 filling the junction region when constant tractions 
$T_i/(\pi r^2)$ act on the four disks that border the cylindrical subsections, and there is no forces on the remaining surface of the junction regions.

Obtaining appropriate trial strain fields is also not too difficult. We first consider an infinitesimal motion that the rod model with \fig{7}(a) as the unit cell can undergo
when the rods are rigid, but the pin junctions are flexible. Then in the final pentamode the junction regions are taken to undergo a rigid body translation which is the same
as that of the corresponding pin junction in the rod model. The cylindrical subsections undergo appropriate shears to ensure continuity of the displacement. The trial displacement
in the remaining multiconnected region of phase 2 bordered by the junction regions and the cylindrical subsection can be somewhat arbitrary, and is not really important. One could
take it as the solution for the displacement field when phase 2 has some nonzero elastic moduli, and the displacement at the boundary of the junctions and cylindrical subsections matches 
that of the trial field just specified. The trial strain field is bounded and non-zero only in phase 2, and therefore the associated upper bound on the elastic energy
scales in proportion to $\Gd$. 

It is clear from the choice of these trial stress fields and trial strain fields, that the macroscopic stress the material supports and the easy motions it permits, are exactly the same
as for the ideal model with rods and pin-junctions that has the unit cell pictured in \fig{7}(a), and which provided the basis for our construction. That this structure can support
any desired average stress, and only that average stress, is then a direct consequence of the analysis in 5.2 of Milton and Cherkaev \cite{Milton:1995:WET}.

\begin{figure}[!ht]
\centering
\includegraphics[width=0.9\textwidth]{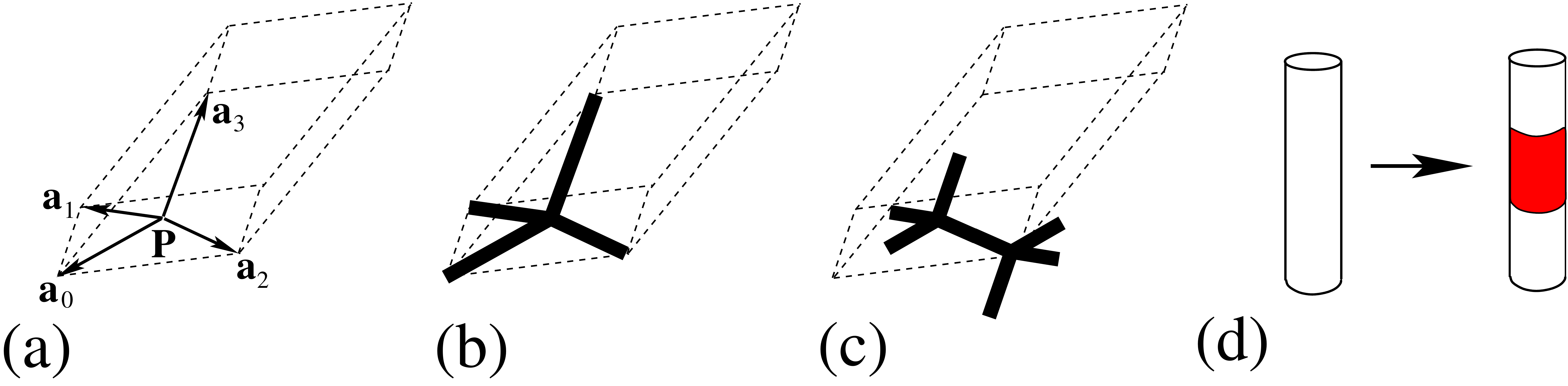}
\caption{The procedure for constructing the desired pentamodes. In (d) a shearable section is inserted into each strut. This section, denoted in red, has the structure of 
parallel square fibers, as illustrated in \fig{8}, with the fibers aligned parallel to the strut. } 
\labfig{7}
\end{figure}
\begin{figure}[!ht]
\vskip 5 mm
\centering
\includegraphics[width=0.7\textwidth]{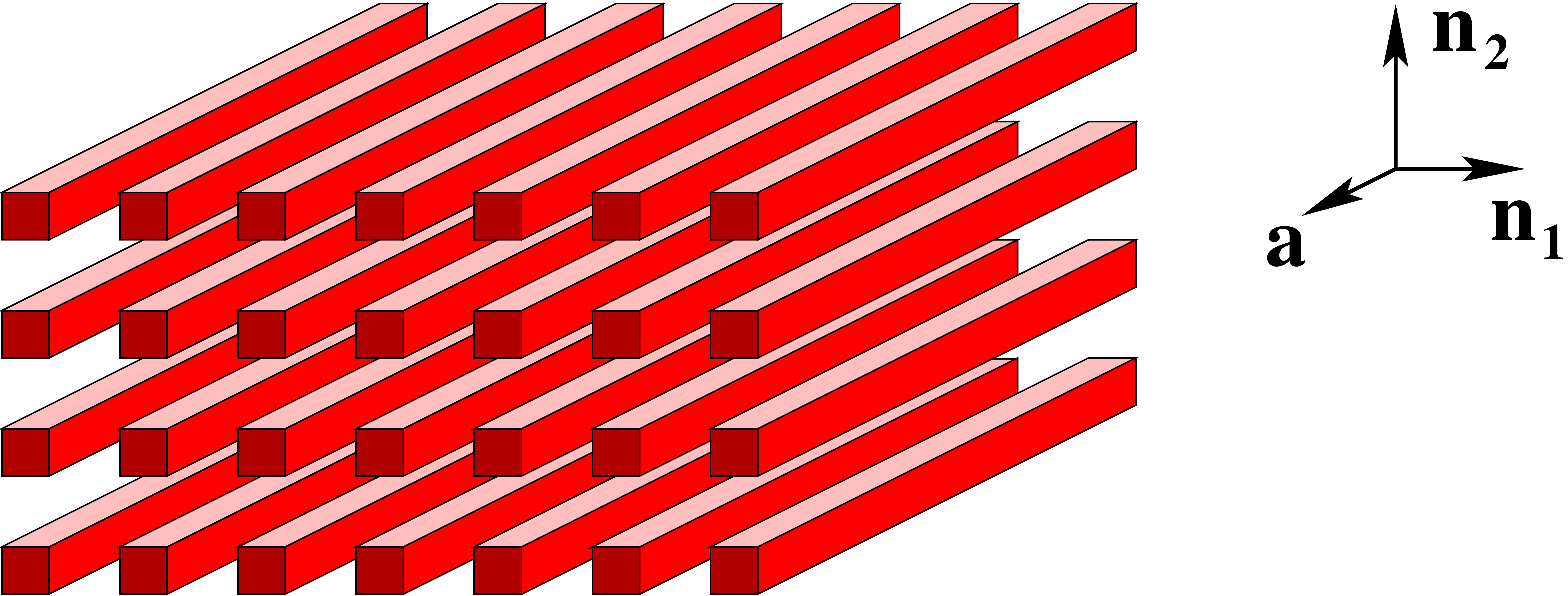}
\caption{A detailed view of the walled pentamode microstructure which is used as the easily shearable section in the pentamode cylindrical struts. The vector $\Ba$ is chosen to be one of the four
vectors $\Ba_k$, $k=0,1,2,3$, as appropriate to each pentamode strut orientated parallel to $\Ba_k$. The square beams can support tension (or compression) in the 
direction of the beam, and in particular it can support a constant macroscopic stress $\BGs_k^B=\Ga_k\Ba_k\Ba_k^T$. As we are working in the framework of linear elasticity
we ignore the very real possibility that the beams will buckle.}
\labfig{8}
\end{figure}
To obtain any desired unimode, bimode, trimode, or quadramode material, that have respectively $p=1,2,3,4$ independent easy modes of deformation, and that
support respectively $6-p$ applied stresses $\BGs_j^0$, $j=1,\ldots,6-p$, we follow the prescription given by Milton and Cherkaev \protect\cite{Milton:1995:WET}. 
That is we superimpose, one at a time, $6-p$ pentamode structures, each supporting one of the stresses $\BGs_j^0$, with struts which are sufficiently thin to ensure that one can (with appropriate modification specified below) superimpose the structures without collision.
When doing this superimposition we first remove phase 2 and 
shift the lattice structures to try to avoid unwanted intersections of phase 1. This may not always be possible, so in the event two vertices clash we make the replacement
in \fig{9}(a) in one of the structures (which may of course then cause additional unwanted intersections of the struts). Then if two (or more) struts intersect we make the replacement \fig{9}(b)
in all but one of the struts (that then passes through each hole). The remaining possibility we want to avoid is that two pentamode struts are parallel and intersect
when we superimpose the structures. Due to the freedom in the choice of the $\Ba_k$, that give a desired $\BGs_j^0$, we can always choose our $6-p$ pentamode structures
to avoid such clashes. Finally, the shearable section in each pentamode strut should be placed in a section that has not been modified,
so it still is parallel to one of the $\Ba_k$. At the very end any remaining space that is not filled by phase 1 should be filled by the extremely compliant phase 2.
\begin{figure}[!ht]
\centering
\includegraphics[width=0.7\textwidth]{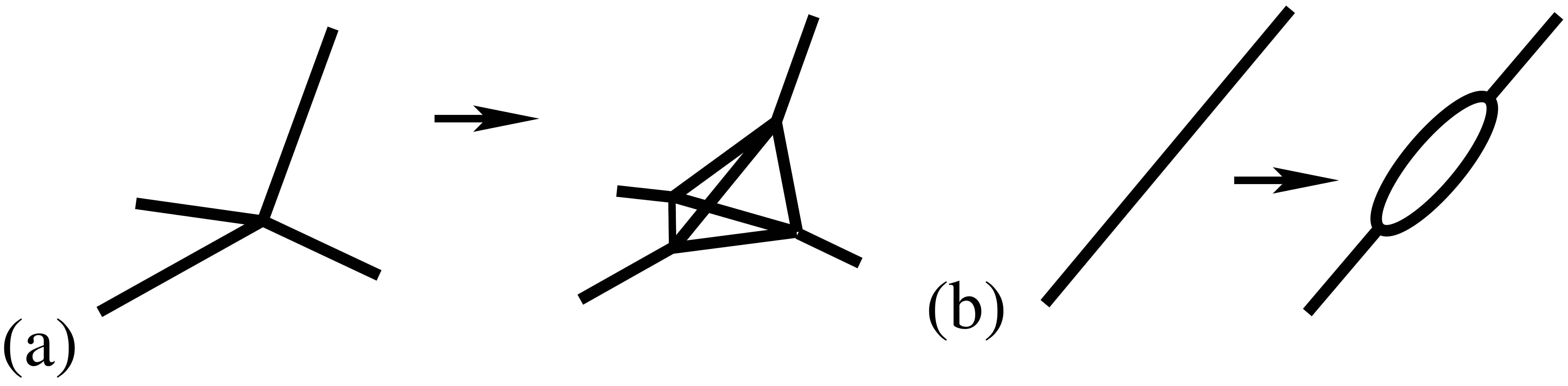}
\caption{Some of the replacements that are needed to obtain desired unimode, bimode, trimode, or quadramode materials.}
\labfig{9}
\end{figure}
\section{Continuity of the energy functions}
\labsect{continuity}
\setcounter{equation}{0}
It follows from the preceding analysis that we can determine the three energy functions
$W_f^3(\BGs^0_1,\BGs^0_2,\BGs^0_3,\BGe^0_1,\BGe^0_2,\BGe^0_3)$, $W_f^4(\BGs^0_1,\BGs^0_2,\BGe^0_1,\BGe^0_2,\BGe^0_3,\BGe^0_4)$, and
$W_f^5(\BGs^0_1,\BGe^0_1,\BGe^0_2,\BGe^0_3,\BGe^0_4,\BGe^0_5)$ in the limit
$\Gd\to 0$ for almost all combinations of applied fields. Here we establish that these energy functions are continuous functions of the
applied fields in the limit
$\Gd\to 0$, and therefore we obtain expressions for the energy functions for all combinations of applied fields in this limit.

Recall that the set $G_fU$ is characterized by its $W$-transform. For example, part of it is described by the function
\beq W_f^4(\BGs^0_1,\BGs^0_2,\BGe^0_1,\BGe^0_2,\BGe^0_3,\BGe^0_4)
 = \min_{\BC_*\in GU_f}\left[\sum_{i=1}^4\BGe^0_i:\BC_*\BGe^0_i+\sum_{j=1}^2\BGs^0_j:\BC_*^{-1}\BGs^0_j\right],
\eeq{6.1}
Here we want to show that such energy functions are continuous in their arguments.
Let the tensor $\BC_*(\BGs^0_1,\BGs^0_2,\BGe^0_1,\BGe^0_2,\BGe^0_3,\BGe^0_4)$ 
be a minimizer of \eq{6.1}, and suppose we perturb the applied stress fields $\BGs^0_j$ by $\Gd\BGs^0_j$,
and the applied strain fields $\BGe^0_i$ by $\Gd\BGe^0_i$. Now consider the following
walled material, with a geometry described by the characteristic function
\beq \chi_w(\Bx)=\prod_{k=1}^{3}(1-H_{\Ge'}(\Bx\cdot\Bn_k)), \eeq{6.2}
where $\Bn_1$, $\Bn_2$, and $\Bn_3$ are the three orthogonal unit vectors,
\beq \Bn_1=\bpm 1 \\ 0 \\ 0 \epm,\quad \Bn_2=\bpm 0 \\ 1 \\ 0 \epm,\quad\Bn_3=\bpm 0 \\ 0 \\ 1 \epm, \eeq{6.3}
and $\Ge'$ is a small parameter that gives the thickness of the walls. Inside the walls, where $\chi_w(\Bx)=0$, we put an
isotropic composite of phase 1 and phase 2, mixed in the proportions $f$ and $1-f$ with isotropic effective
elasticity tensor $\BC(\Gk_0,\Gm_0)$, where $\Gk_0$ is the effective bulk modulus and $\Gm_0$ is the effective 
shear modulus, that are assumed to have non-zero limits as $\Gd\to 0$. (The isotropic composite could consist of islands of void
surrounded by phase 1). Outside the walls, where $\chi_w(\Bx)=1$, we put the material that has an effective tensor 
$\BC_*^{1}=\BC_*(\BGs^0_1,\BGs^0_2,\BGe^0_1,\BGe^0_2,\BGe^0_3,\BGe^0_4)$. Let $\BC_*'$ be the effective tensor of the composite. We have the variational
principle 
\beqa &~&\sum_{i=1}^4(\BGe^0_i+\Gd\BGe^0_i):\BC_*'(\BGe^0_i+\Gd\BGe^0_i)+\sum_{j=1}^2(\BGs^0_j+\Gd\BGs^0_j):(\BC_*')^{-1}(\BGs^0_j+\Gd\BGs^0_j)=\nonum
&~&\min_{\und{\BGe}_1,\und{\BGe}_2,\und{\BGe}_3,\und{\BGe}_4,\und{\BGs}_1,\und{\BGs}_2}
\Big\langle\sum_{i=1}^4\und{\BGe}_i(\Bx):[\chi_w(\Bx)\BC_*^{1}+(1-\chi_w(\Bx))\BC(\Gk_0,\Gm_0)]\und{\BGe}_i(\Bx)\nonum
&~&+\sum_{j=1}^2\und{\BGs}_j(\Bx):[\chi_w(\Bx)\BC_*^{1}+(1-\chi_w(\Bx))\BC(\Gk_0,\Gm_0)]^{-1}\und{\BGs}_j(\Bx)\Big\rangle,
\eeqa{6.4}
where the minimum is over fields subject to the appropriate average values and differential constraints. We choose 
constant trial strain fields
\beq \und{\BGe}_i(\Bx)=\BGe^0_i+\Gd\BGe^0_i,\quad i=1,2,3,4, \eeq{6.5}
and trial stress fields
\beq \und{\BGs}_j(\Bx)=\BGs^0_j+\Gd\und{\BGs}_j(\Bx),\quad j=1,2, \eeq{6.6}
where $\Gd\und{\BGs}_j(\Bx)$ has average value $\Gd\BGs^0_j$ and is concentrated in the walls. Specifically, if $\{\Gd\BGs^0_j\}_{k\ell}$
denote the matrix elements of $\Gd\BGs^0_j$, and letting
\beqa \Gd\BGs_j^1& = &\bpm 0 & 0 & 0 \\ 0 & 0 & \{\Gd\BGs^0_j\}_{23} \\
0 & \{\Gd\BGs^0_j\}_{32} & \{\Gd\BGs^0_j\}_{33} \epm,\nonum
\Gd\BGs_j^2 & = &\bpm \{\Gd\BGs^0_j\}_{11} & 0 & \{\Gd\BGs^0_j\}_{13} \\ 0 & 0 & 0 \\
\{\Gd\BGs^0_j\}_{31} & 0 & 0 \epm, \nonum
\Gd\BGs_j^3 & = &\bpm 0 & \{\Gd\BGs^0_j\}_{12} & 0 \\
\{\Gd\BGs^0_j\}_{21} & \{\Gd\BGs^0_j\}_{22} & 0  \\ 0 & 0 & 0\epm, 
\eeqa{6.7}
then we choose
\beq \Gd\und{\BGs}_j(\Bx)=\sum_{k=1}^{3}\Gd\BGs_j^k H_{\Ge'}(\Bx\cdot\Bn_k)/\Ge', \eeq{6.8}
which has the required average value $\Gd\BGs^0_j$ and satisfies the differential constraints appropriate to a strain field because $\Gd\BGs_j^k\Bn_k=0$. Hence there
exist positive constants $\Ga$ and $\Gb$ such that for sufficiently small $\Ge'$ and for sufficiently small variations $\Gd\BGs^0_j$ and $\Gd\BGe^0_i$ in the applied fields,
we have
\beqa &~&\Big\langle\sum_{i=1}^4\und{\BGe}_i(\Bx):[\chi_w(\Bx)\BC_*^{1}+(1-\chi_w(\Bx))\BC(\Gk_0,\Gm_0)]\und{\BGe}_i(\Bx)\nonum
&~&+\sum_{j=1}^2\und{\BGs}_j(\Bx):[\chi_w(\Bx)\BC_*^{1}+(1-\chi_w(\Bx))\BC(\Gk_0,\Gm_0)]^{-1}\und{\BGs}_j(\Bx)\Big\rangle \nonum
&~&\leq W_f^4(\BGs^0_1,\BGs^0_2,\BGe^0_1,\BGe^0_2,\BGe^0_3,\BGe^0_4)+\Ga\Ge'+\Gb K/\Ge'
\eeqa{6.8a}
where $K$ represents the norm 
\beq K=\sqrt{\sum_{i=1}^4\Gd\BGe^0_i:\Gd\BGe^0_i+\sum_{j=1}^2\Gd\BGs^0_j:\Gd\BGs^0_j}, \eeq{6.9}
of the field variations. Choosing $\Ge'=\sqrt{\Gb K/\Ga}$ to minimize the right hand side of \eq{6.8a} we obtain
\beqa &~&W_f^4(\BGs^0_1+\Gd\BGs^0_1,\BGs^0_2+\Gd\BGs^0_2,\BGe^0_1+\Gd\BGe^0_1,\BGe^0_2+\Gd\BGe^0_2,\BGe^0_3+\Gd\BGe^0_3,\BGe^0_4+\Gd\BGe^0_4)\nonum
&~&\quad\quad\leq W_f^4(\BGs^0_1,\BGs^0_2,\BGe^0_1,\BGe^0_2,\BGe^0_3,\BGe^0_4)+2\sqrt{\Ga\Gb K}.
\eeqa{6.10}
In obtaining the bound \eq{6.8a} we have used the fact that $K^2$ is less than $K$ for sufficiently small $K$, specifically $K<1$.
Clearly the right hand side of \eq{6.10}
approaches $W_f^4(\BGs^0_1,\BGs^0_2,\BGe^0_1,\BGe^0_2,\BGe^0_3,\BGe^0_4)$ as $K\to 0$. On the other hand
by repeating the same argument with the roles of $W_f^4(\BGs^0_1,\BGs^0_2,\BGe^0_1,\BGe^0_2,\BGe^0_3,\BGe^0_4)$ and
$W_f^4(\BGs^0_1+\Gd\BGs^0_1,\BGs^0_2+\Gd\BGs^0_2,\BGe^0_1+\Gd\BGe^0_1,\BGe^0_2+\Gd\BGe^0_2,\BGe^0_3+\Gd\BGe^0_3,\BGe^0_4+\Gd\BGe^0_4)$ reversed,
and with $\BC_*(\BGs^0_1+\Gd\BGs^0_1,\BGs^0_2+\Gd\BGs^0_2,\BGe^0_1+\Gd\BGe^0_1,\BGe^0_2+\Gd\BGe^0_2,\BGe^0_3+\Gd\BGe^0_3,\BGe^0_4+\Gd\BGe^0_4)$
replacing $\BC_*(\BGs^0_1,\BGs^0_2,\BGe^0_1,\BGe^0_2,\BGe^0_3,\BGe^0_4)$ we deduce that
\beqa &~& W_f^4(\BGs^0_1,\BGs^0_2,\BGe^0_1,\BGe^0_2,\BGe^0_3,\BGe^0_4) \nonum
&~&\quad\leq W_f^4(\BGs^0_1+\Gd\BGs^0_1,\BGs^0_2+\Gd\BGs^0_2,\BGe^0_1+\Gd\BGe^0_1,\BGe^0_2+\Gd\BGe^0_2,\BGe^0_3+\Gd\BGe^0_3,\BGe^0_4+\Gd\BGe^0_4)
+2\sqrt{\Ga\Gb K}.\nonum &~&
\eeqa{6.11}
This with \eq{6.10} establishes the continuity of $W_f^4(\BGs^0_1,\BGs^0_2,\BGe^0_1,\BGe^0_2,\BGe^0_3,\BGe^0_4)$. The continuity of the other energy functions follows by
the same argument. 

\section{Conclusion}
\setcounter{equation}{0}
To conclude, we have established the following two Theorems:
\begin{theorem}
Consider composites in three dimensions of two materials with positive definite elasticity tensors $\BC_1$ and $\BC_2=\Gd\BC_0$ mixed in proportions $f$ and $1-f$. Let the seven energy functions $W_f^k$, $k=0,1,\ldots,6$, that characterize the set $G_fU$ (with $U=(\BC_1,\Gd\BC_0)$) of possible elastic tensors
be defined by \eq{2.1}. These energy functions involve a set of applied strains $\BGe^0_i$ and applied stresses $\BGs^0_j$ meeting the orthogonality condition \eq{2.1aa}.
The energy function $W_f^0$ is given by
\beq W_f^0( \BGs^0_1, \BGs^0_2, \BGs^0_3, \BGs^0_4, \BGs^0_5, \BGs^0_6)=\sum_{j=1}^6\BGs^0_j:\widetilde{\BC}_f^A(\BGs^0_1,\BGs^0_2,\BGs^0_3,\BGs^0_4,\BGs^0_5,\BGs^0_6)\BGs^0_j,
\eeq{7.0}
as proved by Avellaneda \cite{Avellaneda:1987:OBM}. Here
$\widetilde{\BC}_f^A(\BGs^0_1,\BGs^0_2,\BGs^0_3,\BGs^0_4,\BGs^0_5,\BGs^0_6)$ is the effective elasticity tensor of a complementary Avellaneda material, that is a sequentially layered laminate with the minimum value of the sum of complementary energies 
\beq \sum_{j=1}^6\BGs^0_j:\BC_*^{-1}\BGs^0_j.
\eeq{7.0a}
Additionally, we now have
\beqa
\lim_{\Gd\to 0}W_f^3(\BGs^0_1,\BGs^0_2,\BGs^0_3,\BGe^0_1,\BGe^0_2,\BGe^0_3) & = & \sum_{j=1}^3\BGs^0_j:[\widetilde{\BC}_f^A(\BGs^0_1,\BGs^0_2,\BGs^0_3,0,0,0)]^{-1}\BGs^0_j ,\nonum
\lim_{\Gd\to 0}W_f^4(\BGs^0_1,\BGs^0_2,\BGe^0_1,\BGe^0_2,\BGe^0_3,\BGe^0_4) & = & \sum_{j=1}^2\BGs^0_j:[\widetilde{\BC}_f^A(\BGs^0_1,\BGs^0_2,0,0,0,0)]^{-1}\BGs^0_j , \nonum
\lim_{\Gd\to 0}W_f^5(\BGs^0_1,\BGe^0_1,\BGe^0_2,\BGe^0_3,\BGe^0_4,\BGe^0_5) & = &  \BGs^0_1:[\widetilde{\BC}_f^A(\BGs^0_1,0,0,0,0,0)]^{-1}\BGs^0_1 , \nonum
\lim_{\Gd\to 0}W_f^6(\BGe^0_1,\BGe^0_2,\BGe^0_3,\BGe^0_4,\BGe^0_5,\BGe^0_6) & = & 0.
\eeqa{7.1}
for all combinations of applied stresses $\BGs_j^0$ and applied strains $\BGe^0_i$. When  $\det\BGe^0_1=0$ but $\BGe^0_1$ is not positive or negative semidefinite, we have
\beq
\lim_{\Gd\to 0}W_f^1(\BGs^0_1,\BGs^0_2,\BGs^0_3,\BGs^0_4,\BGs^0_5,\BGe^0_1)=\sum_{j=1}^5\BGs^0_j:[\widetilde{\BC}_f^A(\BGs^0_1,\BGs^0_2,\BGs^0_3,\BGs^0_4,\BGs^0_5,0)]^{-1}\BGs^0_j,
\eeq{7.2}
while when the equation $\det(\BGe^0_1+t\BGe^0_2)$ has at least two distinct roots for $t$ (the condition for which is given by \eq{3.7}), and additionally, the matrix pencil $\BGe(t)=\BGe^0_1+t\BGe^0_2$ does not contain any positive definite or negative definite matrices as $t$ is varied (which requires that the quantities in
\eq{3.6} are never all positive, or all negative), we have 
\beq 
\lim_{\Gd\to 0} W_f^2(\BGs^0_1,\BGs^0_2,\BGs^0_3,\BGs^0_4,\BGe^0_1,\BGe^0_2)=\sum_{j=1}^4\BGs^0_j:[\widetilde{\BC}_f^A(\BGs^0_1,\BGs^0_2,\BGs^0_3,\BGs^0_4,0,0)]^{-1}\BGs^0_j.
\eeq{7.3}
\end{theorem}

\begin{theorem}
For two-dimensional composites the four energy functions $W_f^k$, $k=0,1,2,3$ are defined by \eq{2.40} and these characterize the set $G_fU$, 
with $U=(\BC_1,\Gd\BC_0)$, of possible elastic tensors $\BC_*$ of composites of two phases with positive definite elasticity tensors $\BC_1$ and
$\BC_2=\Gd\BC_0$. These energy functions involve a set of applied strains $\BGe^0_i$ and applied stresses $\BGs^0_j$ meeting the orthogonality condition \eq{2.1aa}. The energy function $W_f^0$ is given by
\beq W_f^0( \BGs^0_1, \BGs^0_2, \BGs^0_3)=\sum_{j=1}^3\BGs^0_j:\widetilde{\BC}_f^A(\BGs^0_1,\BGs^0_2,\BGs^0_3)\BGs^0_j,
\eeq{7.4}
as proved by Avellaneda \cite{Avellaneda:1987:OBM}, where $\widetilde{\BC}_f^A(\BGs^0_1,\BGs^0_2,\BGs^0_3)$ is the effective elasticity tensor of a 
complementary Avellaneda material, that is a sequentially layered laminate with the minimum value of the sum of complementary energies 
\beq \sum_{j=1}^3\BGs^0_j:\BC_*^{-1}\BGs^0_j.
\eeq{7.4a}
We also have the trivial result that
\beq
\lim_{\Gd\to 0}W_f^3(\BGe^0_1,\BGe^0_2,\BGe^0_3)  =  0.
\eeq{7.5}
When  $\det\BGe^0_1\leq 0$ we have
\beq \lim_{\Gd\to 0}W_f^1(\BGs^0_1,\BGs^0_2,\BGe^0_1)=\sum_{j=1}^2\BGs^0_j:[\widetilde{\BC}_f^A(\BGs^0_1,\BGs^0_2,0)]^{-1}\BGs^0_j,
\eeq{7.6}
while when $\det\BGe^0_1<0$ or when $f(t)=\det(\BGe^0_1+t\BGe^0_2)$  is quadratic in $t$ with two distinct roots, or when $f(t)$ is linear in $t$ with a nonzero $t$ coefficient we have
\beq \lim_{\Gd\to 0}W_f^2(\BGs^0_1,\BGe^0_1,\BGe^0_2)  =   \BGs^0_1:[\widetilde{\BC}_f^A(\BGs^0_1,0,0)]^{-1}\BGs^0_1.
\eeq{7.7}
\end{theorem}

This theorem, and the accompanying microstructures, help define what sort of elastic behaviors are theoretically possible in 2-d and 3-d printed materials.
They should serve as benchmarks for the construction of more realistic microstructures that can be manufactured. He have found the minimum over
all microstructures of various sums of energies and complementary energies. More realistic designs can be obtained by adding to this sum a term
that penalizes the surface area as done for a single energy minimization by Kohn and Wirth \cite{Kohn:2014:OFS,Kohn:2015:OFS}.

It remains an open problem to find expressions for the energy functions in the cases not covered by this theorem. Even for an isotropic composite with
a bulk modulus $\Gk_*$ and a shear modulus $\Gm_*$, the set of all possible pairs $(\Gk_*, \Gm_*)$ is still not completely characterized either in the limit
$\Gd\to 0$ or in the limit $\Gd\to\infty$. In these limits the bounds of Berryman and Milton \cite{Berryman:1988:MRC} 
and Cherkaev and Gibiansky \cite{Cherkaev:1993:CEB} decouple and provide no extra information 
beyond that provided by the Hashin-Shtrikman-Hill bounds \cite{Hashin:1963:VAT, Hashin:1965:EBF, Hill:1963:EPR, Hill:1964:TMP}.  While the results of this paper show that in the limit $\Gd\to 0$  one can obtain two or three-dimensional structures attaining the Hashin-Shtrikman-Hill upper bound on $\Gk_*$, while having $\Gm_*=0$, it is not clear what the maximum value for $\Gm_*$ is, given that $\Gk_*=0$. 

One important corollary of this work is that it gives a complete characterization of the possible triplets $(\BGe^0, \BGs^0, f)$ of average strain, $\BGe^0$, average 
stress $\BGs^0$, and volume fraction $f$, that can occur in 2-d and 3-d printed materials in the limit $\Gd\to 0$. This will be discussed in a separate paper \cite{Milton:2016:CCP}.

\section*{Acknowledgements}
The authors thank the National Science Foundation for support through grant DMS-1211359. M. Briane wishes to thank the Department of Mathematics of the University of Utah for his stay during March 25-April 3 2016. Mohamed Camar-Eddine is thanked for initial collaborations that marked the beginning of this work. The authors are grateful to Martin Wegener and his group for allowing them to use one of their electron micrographs of the three-dimensional pentamode structure, and to Muamer Kadic for help with processing the micrograph.
\bibliographystyle{siamplain}
\bibliography{/Users/milton/tcbook,/Users/milton/newref}
\end{document}